\documentclass[envcountsame,envcountsect,orivec,oribibl]{llncs}
\usepackage{\jobname}
\usepackage{setspace}

\title{%
  Complexity Classifications for Propositional Abduction in Post's Framework%
  \thanks{Supported by ANR {\it Algorithms and complexity} 07-BLAN-0327-04 and DFG grant VO 630/6-1. An earlier version appeared in the
 \emph{Proc. of 12th International Conference on the Principles of Knowledge Representation and Reasoning}, KR'2010, Toronto, Canada.}%
}

\author{%
  Nadia Creignou\inst{1}, Johannes Schmidt\inst{1}, Michael Thomas\inst{2}%
}

\institute{%
  LIF, UMR CNRS 6166, Aix-Marseille Universit\'{e} \newline%
  163, Avenue de Luminy, 13288 Marseille Cedex 9, France\newline%
  \protect\url{creignou@lif.univ-mrs.fr}\newline%
  \protect\url{johannes.schmidt@lif.univ-mrs.fr}\smallskip\newline%
  \and
  Institut f\"ur Theoretische~Informatik, Gottfried Wilhelm Leibniz Universit\"{a}t\newline%
  Appelstr.~4, 30167~Hannover, Germany\newline%
  \protect\url{thomas@thi.uni-hannover.de}\newline
}

\begin{document}

\maketitle

\begin{abstract}
In this paper we investigate the complexity of abduction, a fundamental and important form of non-monotonic reasoning. Given a knowledge base 
explaining the world's behavior it aims at finding an explanation for some observed manifestation. In this paper we  consider
 propositional abduction, where the knowledge base and the manifestation are represented by propositional formulae. 
The problem of deciding whether there exists an explanation has been shown to be \SigPtwo-complete in general. We focus on formulae in 
which the allowed connectives are taken from certain sets of Boolean functions. We consider different variants of the abduction problem 
in restricting both the manifestations and the hypotheses. For all these variants we obtain a complexity classification  for all possible sets of Boolean functions. In this way, we identify easier cases, namely \NP-complete, \coNP-complete and polynomial cases. Thus, we get a detailed picture of the complexity of the propositional abduction problem, hence highlighting sources of intractability. Further, we address the problem of counting the explanations and draw a complete picture for the counting complexity.
\newline

\end{abstract}

\newpage
\section{Introduction}
This paper is dedicated to the computational complexity of    abduction, a fundamental and important form of non-monotonic reasoning. Given a  certain consistent knowledge about the world, abductive reasoning is used to generate explanations (or at least telling if there is one) for observed manifestations. Nowadays abduction has taken on fundamental importance in Artificial Intelligence and has  many application areas spanning medical diagnosis~\cite{byaltajo89}, text analysis~\cite{hostapma93}, system diagnosis~\cite{stwo01},
configuration problems~\cite{amfama02}, temporal knowledge bases~\cite{boli00} and has connections to default reasoning~\cite{selman-levesque:1990}.
  
  There are several approaches to formalize the problem of abduction.
  In this paper, we focus on \emph{logic based abduction} in which the knowledge base
  is given as a set $\Gamma$ of propositional formulae. We are interested in deciding whether there exists an \emph{explanation} $E$, \emph{i.e.},
  a set of literals consistent with $\Gamma$ such that $\Gamma$ and $E$ together entail the observation. 

  From a complexity theoretic viewpoint, the abduction problem is very hard since it is $\SigPtwo$-complete and thus situated at the
  second level of the polynomial hierarchy~\cite{eiter-gottlob:1995}. This intractability result raises the question for restrictions leading to
  fragments of lower complexity. Several such restrictions have been considered in previous works. One of the most famous amongst those is
  Schaefer's framework, where formulae are restricted to generalized conjunctive normal form with clauses from a fixed set of
  relations~\cite{crza06,noza05,noza08}.
  
  A similar yet different procedure is to rather require formulae to be constructed from a restricted set of Boolean functions $B$.
  Such formulae are called \emph{$B$-formulae}.
  This approach has first been taken by Lewis, who showed that the satisfiability problem is $\NP$-complete if and only if this set of
  Boolean functions has the ability to express the negation of implication connective $\not\limplies$~\cite{lew79}. 
  Since then, this approach has been applied to a wide range of problems including
  equivalence and implication problems~\cite{rei03,bemethvo08imp},
  satisfiability and model checking in modal and temporal logics~\cite{bhss05,bsssv07},
  default logic~\cite{bemethvo08}, and circumscription~\cite{thomas09}, among others.

  We follow this approach and show that Post's lattice allows to completely classify the complexity of propositional abduction 
  for several variants and all possible sets of allowed Boolean functions.
  We consider two main variants of the abduction problem. In the first one we may build explanations from
  positive and negative literals. We refer to this problem as \emph{symmetric} abduction, $\ABD$ for short.
  The second variant, $\ABD[\HP]$, is the so-called \emph{positive} abduction where we allow only positive literals in the explanations.

  We first examine the symmetric variant in the case where the representation of the manifestation is a positive literal.
  We show that depending on the set $B$ of allowed connectives the abduction problem is either $\SigPtwo$-complete, or $\NP$-complete, or
  in $\P$ and $\ParityL$-hard, or in  $\L$. More precisely, we prove that the complexity of this abduction problem is $\SigPtwo$-complete
  as soon as $B$ can express one of the functions $x \vee (y \wedge \neg z)$, $x \wedge (y \vee \neg z)$ or
  $(x \wedge y) \vee (x \wedge \neg z) \vee (y \wedge \neg z)$. It drops to $\NP$-complete when all functions in $B$ are monotonic and have
  the ability to express one of the functions $x \vee (y \wedge z)$, $x \wedge (y \vee z)$ or
  $(x \wedge y) \vee (x \wedge z) \vee (y \wedge z)$. 
  The problem becomes solvable in polynomial time and is $\ParityL$-hard if $B$-formulae may depend
  on more than one variable while being representable as linear equations. 
  Finally the complexity drops to $\L$ in all remaining cases.
  We then complete our study of symmetric abduction with analogous complexity classifications of the 
  variants of $\ABD$ obtained by restricting the manifestation to be respectively a clause, a term or a $B$-formula.
  
  These results are subsequently extended to positive abduction.
  An overview can be found in Figures~\ref{fig:abduction_complexity} 
  and~\ref{fig:positive_abduction_complexity}.

  Please note that in \cite{crza06}, 
  the authors obtained a complexity classification of the abduction problem in the so-called Schaefer's framework. 
  The two classifications are in the same vein since they classify the complexity of abduction for local restrictions on the knowledge base. 
  However the two results are incomparable in the sense that no classification can be deduced from the other. 
  They only overlap in the particular case of the linear connective $\xor$, for which both types of sets of formulae can be
  seen as systems of linear equations. This special abduction case has been shown to be decidable in polynomial time in \cite{zanuttini03}.

  Besides the decision problem, another natural question is concerned with the number of explanations. This problem refers to the counting
  problem for abduction. The study of the counting complexity of abduction has been started by Hermann and Pichler (\cite{hepi07}).
  We prove here a trichotomy theorem showing that counting the full explanations of symmetric abduction is either $\SHcoNP$-complete or
  $\SHP$-complete or in $\FP$, depending on the set $B$ of allowed connectives. 
  We also consider the counting problem associated with positive abduction, for which we 
  distinguish two frequently used settings: counting either all positive explanations, or counting the subset-minimal.
  For both formalizations of the counting problem, we get a potentially dichotomous classification with one open case.

  The rest of the paper is structured as follows. 
  We first give the necessary preliminaries in Section~\ref{sec:preliminaries}.
  The abduction problems considered herein are defined in Section~\ref{sec:abduction_problem}.
  In Section~\ref{sec:complexity_abduction} we classify the complexity of symmetric abduction, 
  where 
    we first consider the case where the manifestation is  a single positive literal (Section~\ref{subsec:complexity_pq}), 
    and then turn to variants in which the manifestations are clauses, terms and restricted formulae (Section~\ref{subsect:abduction-variants}). 
  Section~\ref{sec:positive_abduction} then studies the complexity of positive abduction. 
  An overview of these results is given in Section~\ref{sec:overview}. 
  Finally, Section~\ref{sec:counting} is dedicated to the counting problem, 
  and Section~\ref{sec:conclusion} contains some concluding remarks.


\section{Preliminaries}\label{sec:preliminaries}
\paragraph{Complexity Theory}

  We require standard notions of complexity theory. 
  For the decision problems the arising complexity degrees encompass the classes $\L$, $\P$, $\NP$, and $\SigPtwo$. 
  For more background information, the reader is referred to~\cite{pap94}. 
  We furthermore require the class $\ParityL$ defined as the class of languages $L$ such that there exists a nondeterministic logspace Turing machine
  that exhibits an odd number of accepting paths if and only if $x\in L$, for all $x$~\cite{budaheme92}. 
  It holds that $\L \subseteq \ParityL \subseteq \P$.
  For our hardness results we employ \emph{logspace many-one reductions}, defined as follows:
  a language $A$ is logspace many-one reducible to some language $B$ (written $A \leqlogm B$) if
  there exists a logspace-computable function $f$ such that $x \in A$ if and only if  $f(x) \in B$.

  A \emph{counting problem} is represented using a \emph{witness function} $w$, which for every input $x$ returns a finite set of witnesses. This witness function gives rise to the following counting problem: given an instance $x$, find the cardinality $\vert w(x)\vert$ of the witness set $w(x)$. The class $\SHP$ is the class of counting problems naturally associated with decision problems in $\NP$. According to \cite{hevo95}  if ${\cal C}$ is a complexity class of decision problems, we define $\SHclass{\cal C}$ to be the class of all counting problems whose witness function is such that the size of every witness $y$ of $x$ is polynomially bounded in the size of $x$, and checking  whether $y\in w(x)$ is in ${\cal C}$. Thus, we have $\SHP=\SHclass{\P}$ and $\SHP\subseteq \SHcoNP$. Completeness of counting problems is usually proved by means of Turing reductions. 
A stronger notion is the parsimonious reduction where the exact number of solutions is conserved by the reduction function.

\paragraph{Propositional formulae}

  We assume familiarity with propositional logic. 
  The set of all propositional formulae is denoted by $\allFormulae$.
  A \emph{model} for a formula $\varphi$ is a truth assignment to the set of its variables that satisfies $\varphi$.
  Further we denote by $\varphi[\alpha/\beta]$ the formula obtained from $\varphi$ by replacing all occurrences of $\alpha$ with $\beta$.
  For a given set $\Gamma$ of formulae, we write $\Vars{\Gamma}$ to denote the set of variables occurring in $\Gamma$.
  We identify finite $\Gamma$ with the conjunction of all the formulae in $\Gamma$, $\bigwedge_{\varphi\in\Gamma}\varphi$. Naturally,
  $\Gamma[\alpha/\beta]$ then stands for $\bigwedge_{\varphi\in\Gamma}\varphi[\alpha/\beta]$.
  For any formula $\varphi \in \allFormulae$, we write $\Gamma \models \varphi$ if $\Gamma$ entails $\varphi$, 
  \emph{i.e.}, if every model of $\Gamma$ also satisfies $\varphi$. 

  A \emph{literal} $l$ is a variable $x$ or its negation $\neg x$.
  Given a set of variables $V$, $\Lits{V}$ denotes the set of all literals formed upon the variables in $V$, \emph{i.e.},
  $\Lits{V} := V \cup \{\neg x \mid x \in V\}$.
  A \emph{clause} is a disjunction of literals 
  and a \emph{term} is a conjunction of literals.

\paragraph{Clones of Boolean Functions}

  A \emph{clone} is a set of Boolean functions that is closed under superposition, \emph{i.e.}, it contains all projections 
  (that is, the functions $f(a_1, \dots , a_n) = a_k$ for all $1 \leq k \leq n$ and $n \in \N$) 
  and is closed under arbitrary composition.
  Let $B$ be a finite set of Boolean functions. 
  We denote by $[B]$ the smallest clone containing $B$ and call $B$ a \emph{base} for $[B]$.
  In 1941 Post identified the set of all clones of Boolean functions \cite{pos41}.
  He gave a finite base for each of the clones and showed that they form a lattice under the usual $\subseteq$-relation, 
  hence the name \emph{Post's lattice} (see, \emph{e.g.}, Figure~\ref{fig:abduction_complexity}).
  To define the clones we introduce the following notions, where $f$ is an $n$-ary Boolean function:
  \begin{itemize} \itemsep 0pt 
    \item $f$ is \emph{$c$-reproducing} if $f(c, \ldots , c) = c$, $c \in \{\false,\true\}$.
    \item $f$ is \emph{monotonic} if $a_1 \leq b_1, \ldots , a_n \leq b_n$ implies $f(a_1, \ldots , a_n) \leq f(b_1, \ldots , b_n)$.
    \item $f$ is \emph{$c$-separating of degree $k$} if 
    for all $A \subseteq f^{-1}(c)$ of size $|A|=k$ 
    there exists an $i \in \{1, \ldots , n\}$ 
    such that $(a_1, \ldots , a_n) \in A$ implies $a_i = c$, $c \in \{\false,\true\}$.
    \item $f$ is \emph{$c$-separating} if 
    $f$ is $c$-separating of degree $|f^{-1}(c)|$.
    \item $f$ is \emph{self-dual} if $f \equiv \mathrm{dual}(f)$, where $\mathrm{dual}(f)(x_1,\ldots,x_n) := \neg f(\neg x_1, \ldots , \neg x_n)$.
    \item $f$ is \emph{affine} if $f \equiv x_1 \xor \cdots \xor x_n \xor c$ with $c \in \{0, 1\}$.
  \end{itemize}
  
  A list of all clones with definitions and finite bases is given in Table~\ref{tab:clones} on page \pageref{tab:clones}.
  A propositional formula using only functions from $B$ as connectives is called a \emph{$B$-formula}.
  The set of all $B$-formulae is denoted by $\allFormulae(B)$.
  
  Let $f$ be an $n$-ary Boolean function. A $B$-formula $\varphi$ such that $\Vars{\varphi}= \{x_1,\ldots , x_n, y_1,\ldots , y_m\}$ is a
  \emph{$B$-representation} of $f$ if for all $a_1,\ldots, a_n, b_1,\dots , b_m \in \{\false,\true\}$ it holds that $f(a_1,\ldots , a_n)=1$
  if and only if every $\sigma\colon\Vars{\varphi}\longrightarrow\{\false,\true\}$ with $\sigma(x_i)=a_i$ and $\sigma(y_i)=b_i$ for all relevant $i$
  satisfies $\varphi$. Such a $B$-representation exists for every $f\in[B]$. Yet, it may happen that the $B$-representation of some function uses
  some input variable more than once.
  \begin{example}
    Let $h(x,y)=x\wedge \neg y$. An $\{h\}$-representation of the function $x\land y$ is $h(x,h(x,y))$.
  \end{example}

  \begin{table*}[tb]
    \centering
    \fontsize{8pt}{10.4pt}\selectfont
    \begin{tabular}{p{1cm}ll}
      \specialrule{\heavyrulewidth}{0cm}{0cm}
      Name & Definition & Base \\
      \specialrule{\heavyrulewidth}{0cm}{0.03cm}
      $\CloneBF$ & All Boolean functions & $\{x \land y, \neg x\}$ \\
      \hline
      $\CloneR_0$ & $\{f \mid f \text{ is $\false$-reproducing}\}$ & $\{x \land y, x \xor y\}$ \\
      \hline
      $\CloneR_1$ & $\{f \mid f \text{ is $\true$-reproducing}\}$ & $\{x \lor y, x \xor y \xor \true \}$ \\
      \hline
      $\CloneR_2$ & $\CloneR_0 \cap \CloneR_1$ & $\{\lor, x \land (y \xor z \xor \true) \}$ \\
      \hline
      $\CloneM$ & $\{f \mid f \text{ is monotonic}\}$ & $\{x \lor y, x \land y, \false, \true\}$ \\
      \hline
      $\CloneM_1$ & $\CloneM \cap \CloneR_1$ & $\{x \lor y, x \land y, \true\}$ \\
      \hline
      $\CloneM_0$ & $\CloneM \cap \CloneR_0$ & $\{x \lor y, x \land y, \false\}$ \\
      \hline
      $\CloneM_2$ & $\CloneM \cap \CloneR_2$ & $\{x \lor y, x \land y \}$ \\
      \hline

      $\CloneS^n_{0}$ & $\{f \mid f \text{ is $0$-separating of degree } n\}$ & $\{x \to y, \dual{h_n}\}$ \\
      \hline
      $\CloneS_0$ & $\{f \mid f \text{ is $0$-separating}\}$ & $\{x \to y\}$ \\
      \hline
      $\CloneS^n_{1}$ & $\{f \mid f \text{ is $1$-separating of degree } n\}$ & $\{x \wedge \neg y, h_n\}$ \\
      \hline
      $\CloneS_1$ & $\{f \mid f \text{ is $1$-separating}\}$ & $\{x \wedge \neg y\}$ \\
      \hline
      $\CloneS^n_{02}$ & $\CloneS^n_0 \cap \CloneR_2$ & $\{x \lor  (y  \land  \neg z), \dual{h_n}\}$ \\
      \hline
      $\CloneS_{02}$ & $\CloneS_0 \cap \CloneR_2$ & $\{x \lor  (y  \land  \neg z)\}$ \\
      \hline
      $\CloneS^n_{01}$ & $\CloneS^n_0 \cap \CloneM$ & $\{\dual{h_n},\true\}$ \\
      \hline
      $\CloneS_{01}$ & $\CloneS_0 \cap \CloneM$ & $\{x \lor  (y  \land z), \true\}$ \\
      \hline
      $\CloneS^n_{00}$ & $\CloneS^n_0 \cap \CloneR_2 \cap \CloneM$ & $\{x \lor  (y  \land  z), \dual{h_n}\}$ \\
      \hline
      $\CloneS_{00}$ & $\CloneS_0 \cap \CloneR_2 \cap \CloneM$ & $\{x \lor  (y  \land  z)\}$ \\
      \hline
      $\CloneS^n_{12}$ & $\CloneS^n_1 \cap \CloneR_2$ & $\{x \land  (y  \lor  \neg z), h_n\}$ \\
      \hline
      $\CloneS_{12}$ & $\CloneS_1 \cap \CloneR_2$ & $\{x \land  (y  \lor  \neg z)\}$ \\
      \hline
      $\CloneS^n_{11}$ & $\CloneS^n_1 \cap \CloneM$ & $\{h_n, \false\}$ \\
      \hline
      $\CloneS_{11}$ & $\CloneS_1 \cap \CloneM$ & $\{x \land  (y  \lor  z), \false\}$ \\
      \hline
      $\CloneS^n_{10}$ & $\CloneS^n_1 \cap \CloneR_2 \cap \CloneM$ & $\{x \land  (y  \lor  z), h_n\}$ \\
      \hline
      $\CloneS_{10}$ & $\CloneS_1 \cap \CloneR_2 \cap \CloneM$ & $\{x \land  (y  \lor  z)\}$ \\
      \hline

      $\CloneD$ & $\{f \mid f \text{ is self-dual}\}$ & $ \{(x \land \neg y) \lor (x \land \neg z) \lor (\neg y \land  \neg z)\}$ \\
      \hline
      $\CloneD_1$ & $\CloneD \cap \CloneR_2$ & $ \{(x \land y) \lor (x \land \neg z) \lor (y \land \neg z)\}$ \\
      \hline
      $\CloneD_2$ & $\CloneD \cap \CloneM$ & $ \{(x \land y) \lor (x \land z) \lor (y \land z)\}$ \\
      \hline

      $\CloneL$ & $\{f \mid f \text{ is affine}\}$ & $\{ x \xor y,\true\}$ \\
      \hline
      $\CloneL_0$ & $\CloneL \cap \CloneR_\false$ & $\{x \xor y \}$ \\
      \hline
      $\CloneL_1$ & $\CloneL \cap \CloneR_\true$ & $\{x \xor y  \xor \true \}$ \\
      \hline
      $\CloneL_2$ & $\CloneL \cap \CloneR_2$ & $\{x \xor y \xor z\}$ \\
      \hline
      $\CloneL_3$ & $\CloneL \cap \CloneD$ & $\{x  \xor  y  \xor  z  \xor  \true\}$ \\
      \hline

      $\CloneV$ & $\{f \mid  f $ is a disjunction of variables or constants$\}$ & $\{ x \lor y, \false,\true \}$ \\
      \hline
      $\CloneV_0$ & $\CloneV \cap \CloneR_0$ & $\{ x \lor y, \false\}$ \\
      \hline
      $\CloneV_1$ & $\CloneV \cap \CloneR_1$ & $\{ x \lor y, \true\}$ \\
      \hline
      $\CloneV_2$ & $\CloneV \cap \CloneR_2$ & $\{ x \lor y\}$ \\
      \hline

      $\CloneE$ & $\{f \mid  f $ is a conjunction of variables or constants$\}$ & $\{ x \land y, \false, \true \}$ \\
      \hline
      $\CloneE_0$ & $\CloneE \cap \CloneR_0$ & $\{ x \land y, \false\}$ \\
      \hline
      $\CloneE_1$ & $\CloneE \cap \CloneR_1$ & $\{ x \land y, \true\}$ \\
      \hline
      $\CloneE_2$ & $\CloneE \cap \CloneR_2$ & $\{ x \land y\}$ \\
      \hline

      $\CloneN$ & $\{f \mid f $ depends on at most one variable$\}$ & $\{ \neg x,\false,\true\}$ \\
      \hline
      $\CloneN_2$ & $\CloneN \cap \CloneR_2$ & $\{ \neg x\}$ \\
      \hline

      $\CloneI$ & $\{f \mid f \text{ is a projection or a constant}\}$ & $\{\id, \false,\true\}$ \\
      \hline
      $\CloneI_0$ & $\CloneI \cap \CloneR_0$ & $\{\id, \false\}$ \\
      \hline
      $\CloneI_1$ & $\CloneI \cap \CloneR_1$ & $\{\id, \true\}$ \\
      \hline
      $\CloneI_2$ & $\CloneI \cap \CloneR_2$ & $\{\id\}$ \\
      \specialrule{\heavyrulewidth}{0cm}{0cm}
    \end{tabular}
    \smallskip
    \caption{\label{tab:clones}%
      The list of all Boolean clones with definitions and bases, where $h_n := \bigvee^{n+1}_{i=1}\bigwedge^{n+1}_{j=1,j\neq i} x_j$ and
      $\dual{f}(a_1, \dots , a_n) = \neg f(\neg a_1 \dots , \neg a_n)$.
    }
  \end{table*}
  
  \begin{figure*}[t]
    \centering
    \includegraphics[width =\textwidth]{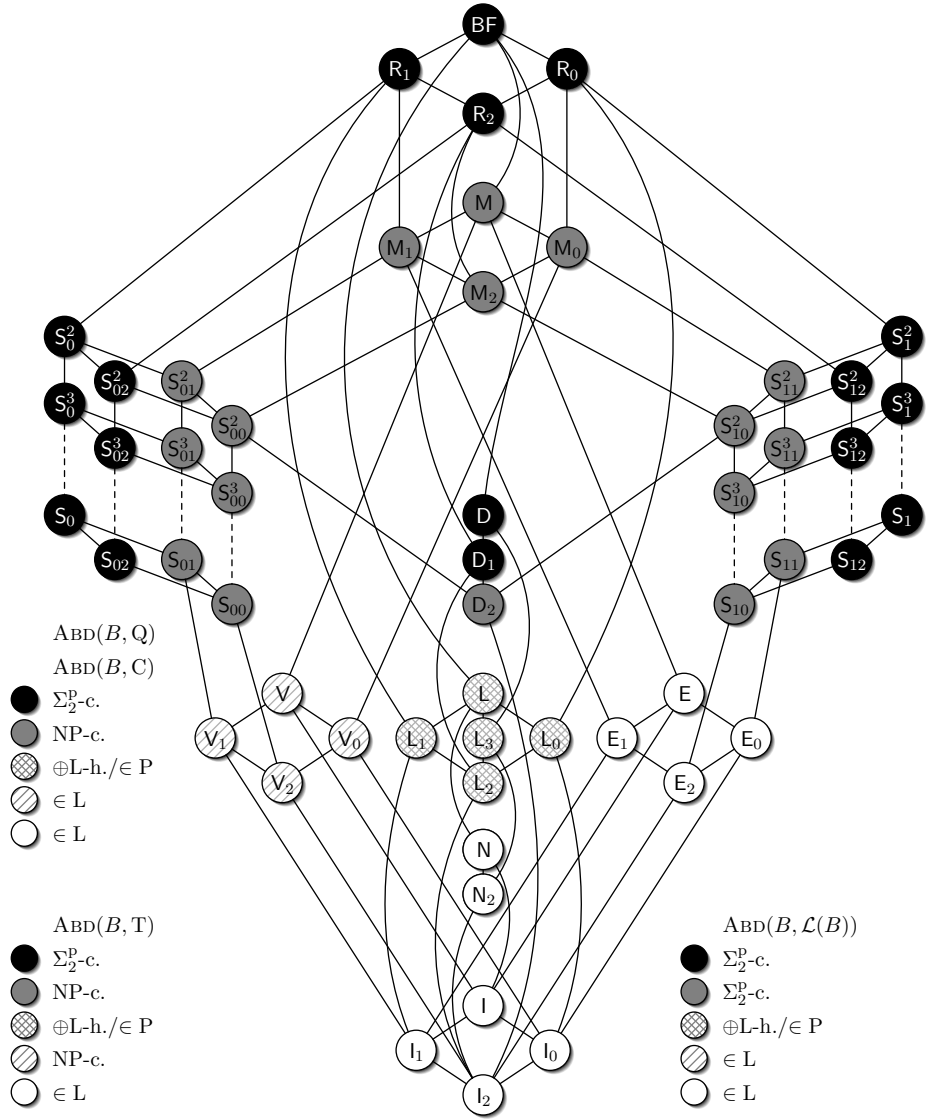}
    \caption{\label{fig:abduction_complexity}Post's lattice showing the complexity of the symmetric abduction problem $\ABD(B,\manif)$ for all sets $B$
      of Boolean functions and the most interesting restrictions $\manif$ of the manifestations. In the legend, $\sL$ abbreviates $\L$ and the suffixes ``-h'' and ``-c'' indicate hardness and completeness respectively.}
  \end{figure*}

  \begin{figure*}[t]
    \centering
    \includegraphics[width =\textwidth]{p-abduction.0} 
    \caption{\label{fig:positive_abduction_complexity} Post's lattice showing the complexity of the positive abduction problem $\ABD[\HP](B,\manif)$
    for all sets $B$ of Boolean functions and the most interesting restrictions $\manif$ of the manifestations. In the legend, $\sL$ abbreviates $\L$ and the suffixes ``-h'' and ``-c'' indicate hardness and completeness respectively.}
  \end{figure*}
\section{The Abduction Problem}\label{sec:abduction_problem}

  Let $B$ be a finite set of Boolean functions. We are interested in the propositional abduction problem
  parameterized by the set $B$ of allowed connectives. 
  We define the \emph{abduction problem for $B$-formulae} as
  \problemdef{$\ABD(B)$}
      {$\calP=(\Gamma,A,\varphi)$, where
      \begin{itemize}
        \item $\Gamma\subseteq\allFormulae(B)$ is a set of $B$-formulae,
        \item $A\subseteq \Vars{\Gamma}$ is a set of variables,
        \item $\varphi \in \allFormulae$ is a formula with $\Vars{\varphi} \subseteq \Vars{\Gamma}\setminus A$.
      \end{itemize}
      }
      {Is there a set $E \subseteq \Lits{A}$ such that $\Gamma \land E$ is satisfiable and $\Gamma \land E \models \varphi$ (or equivalently $\Gamma
      \land E \land \neg\varphi$ is unsatisfiable)?}
  The set $\Gamma$ represents the \emph{knowledge base}. The set $A$ is called the set of \emph{hypotheses} and $\varphi$ is called
  \emph{manifestation} or \emph{query}.
  Furthermore, if such a set $E$ exists, it is called an \emph{explanation} or a \emph{solution} of the abduction problem.
  It is called  a \emph{full explanation} if $\Vars{E}=A$. Observe that every explanation can be extended to a full one.
  We will consider several restrictions of the manifestations of this problem.
  To indicate them, we introduce a second argument $\manif$ meaning that $\varphi$ is required to be
  \begin{itemize}
    \item $\MQ$ (resp.\ $\MPQ$, $\MNQ$): a single literal (resp.\ positive literal, negative literal),
    \item $\MC$ (resp.\ $\MPC$, $\MNC$): a clause (resp.\ positive clause, negative clause),
    \item $\MT$ (resp.\ $\MPT$, $\MNT$): a term (resp.\ positive term, negative term),
    \item $\MF$: a $B$-formula.
  \end{itemize}
  We refer to the above defined abduction problem as \emph{symmetric} abduction, since every variable of the hypotheses $A$ may be taken positive or
  negative to construct an explanation.
  We will also consider \emph{positive} abduction, where we are interested in purely positive explanations only. To indicate this, we
  add the prefix ``$\HP$-''. Thus, for an instance $(\Gamma, A, \varphi)$ of $\ABD[\HP](B,\manif)$, every solution $E$ of $(\Gamma, A, \varphi)$
  has to satisfy $E \subseteq A$.

  The following important lemma makes clear the role of the constants in our abduction problem. 
  It often reduces the number of cases to be considered.
  
  \begin{lemma} \label{lem:constants_sometimes_available}
  Let $B$ be a finite set of Boolean functions.
  \begin{enumerate}
    \item If $\manif \in \{\MQ, \MC, \MT, \MF\}$, then
      \begin{align*}
          \ABD(B,\manif)			& \equivlogm \ABD  (B \cup \{\true\},\manif), \\
          \ABD[\HP](B,\manif) &	\equivlogm \ABD[\HP](B \cup \{\true\},\manif).
      \end{align*}
    \item If $\manif \in \{\MQ, \MC, \MT\}$ and $\lor\in [B]$, then 
          \[\ABD(B,\manif) \equivlogm \ABD(B \cup \{\false\},\manif).\]
  \end{enumerate}
  \end{lemma}

  \begin{proof}
    To reduce $\ABD(B \cup \{\true\},\manif)$ to $\ABD(B,\manif)$ we transform any instance of the first problem in replacing every
    occurrence of $\true$ by a fresh variable $t$ and adding the unit clause $(t)$ to the knowledge base. The same reduction works for $\ABD[\HP]$.
    To prove $\ABD(B \cup \{\false\},\manif)\equivlogm\ABD(B,\manif) $, let $\calP=(\Gamma,A,\psi)$  be an instance of the first problem
    and $f$ be a fresh variable. Since $\manif \in \{\MQ,\MC,\MT\}$, we can suppose w.l.o.g.\ that $\psi$ does not contain $\false$. We map $\calP$ to
    $\calP'=(\Gamma', A', \psi)$, where $\Gamma'$ is the $B$-representation of $\{\varphi[\false/f]\lor f\mid \varphi\in\Gamma\}$ and
    $A' = A \cup \{f\}$.
    \qed
  \end{proof}

  Of course this lemma holds also for purely positive/negative queries, clauses or terms, \emph{i.e.}, $\MQ,\MC,\MT$ can be replaced by
  $\MPQ,\MPC,\MPT$ or $\MNQ,\MNC,\MNT$, respectively.
  
  Observe that if $B_1$ and $B_2$ are two sets of Boolean functions such that $B_1\subseteq [B_2]$, then every function of $B_1$ can be expressed by
  a $B_2$-formula,  namely by its \emph{$B_2$-representation}. This way there is a canonical reduction between $\ABD(B_1,\manif)$ and
  $\ABD(B_2, \manif)$ if $B_1\subseteq [B_2]$: replace all $B_1$-connectives by their $B_2$-representa\-tion. Note that this reduction is not
  necessarily polynomial: Since the $B_2$-representation of some function may use some input variable more than once, the formula size may grow
  exponentially. Nevertheless we will use this reduction very frequently, avoiding an exponential blow-up by special structures of the $B_1$-formulae.

\section{The complexity of $\ABD(B)$}\label{sec:complexity_abduction}
  We commence with the symmetric abduction problem $\ABD(B)$
  The results of this section are summarized in Figure \ref{fig:abduction_complexity}.
  We will first focus on the case where $\varphi$ is a single positive literal, thus discussing the problem $\ABD(B,\MPQ)$.

\subsection{The complexity of $\ABD(B,\MPQ)$}\label{subsec:complexity_pq}
  
  \begin{theorem}\label{thm:main}
    Let $B$ be a finite set of Boolean functions. Then the symmetric abduction problem for propositional $B$-formulae with a positive literal manifestation, $\ABD(B,\MPQ)$,  is
    \begin{enumerate}
      \item $\SigPtwo$-complete if $\CloneS_{02}\subseteq [B]$ or $\CloneS_{12}\subseteq [B]$ or  $\CloneD_{1}\subseteq [B]$,
      \item $\NP$-complete if  $\CloneS_{00}\subseteq [B]\subseteq \CloneM$ or $\CloneS_{10}\subseteq [B]\subseteq \CloneM$ or $\CloneD_{2}
              \subseteq [B]  \subseteq \CloneM$,
      \item in $\P$ and $\ParityL$-hard if $\CloneL_{2}\subseteq [B]\subseteq \CloneL$, and
      \item in $\L$ in all other cases.
    \end{enumerate}
  \end{theorem}

  \begin{remark}\label{rem:recognize_the_clone}
    For such a classification a natural question is: given $B$, how hard is it to determine the complexity of $\ABD(B,\MPQ)$? Solving this
    task requires checking whether certain clones are included in $[B]$ (for lower bounds) and whether $B$ itself is included in certain clones
    (for upper bounds). As shown in \cite{vol09}, the complexity of checking whether certain Boolean functions are included in a clone depends on the
    representation of the Boolean functions. If all functions are
    given by their truth table then the problem is in quasi-polynomial-size $\AC{0}$, while if the input functions are given in a compact way,
    \emph{i.e.}, by circuits, then the above problem becomes $\coNP$-complete.
  \end{remark}

  We split the proof of Theorem~\ref{thm:main} into several propositions.

  \begin{proposition}\label{prop:abd(PQ)-E-N-V-logspace}
    Let $B$ be a finite set of Boolean functions such that $[B] \subseteq \CloneE$ or $[B] \subseteq \CloneN$ or $[B] \subseteq \CloneV$.
    Then $\ABD(B,\MPQ) \in \L$.
  \end{proposition}
  
  \begin{proof}
    Let $\calP=(\Gamma,A,q)$ be an instance of $\ABD(B,\MPQ)$.
    
    For $[B] = \CloneN$ or $\CloneE$, $\Gamma$ is equivalent to a set of literals,
    hence $\calP$ has the empty set as a solution if $\calP$ possesses a solution at all.
    Finally notice that satisfiability of a set of $\CloneN$-formulae can be tested in logarithmic space \cite{sch05}.
    
    For $[B] = \CloneV$ each formula $\varphi \in \Gamma$ is equivalent to either a constant or disjunction.
    It holds that  $(\Gamma,A,q)$ has a solution if and only if  $\Gamma$ contains a formula
    $\varphi \equiv q \lor x_1 \lor \cdots \lor x_k$ such that $\{x_1, \dots x_k\} \subseteq A$, and $\Gamma[x_1/\false,\ldots,x_k/\false]$ is satisfiable. This can be
    tested in logarithmic space, as substitution of symbols and  evaluation of $\CloneV$-formulae can all be performed in logarithmic space.
    \qed
  \end{proof}

  \begin{proposition}\label{prop:abd(PQ)-L-P}
    Let $B$ be a finite set of Boolean functions such that $\CloneL_2 \subseteq [B] \subseteq \CloneL$. 
    Then $\ABD(B,\MPQ)$ is $\ParityL$-hard and contained in $\P$.
  \end{proposition}

  \begin{proof}
    In this case, deciding whether an instance of $\ABD(B,\MPQ)$ has a solution logspace reduces to the problem of deciding whether a propositional
    abduction problem in which the knowledge base is a set of linear equations has a solution. This has been shown to be decidable in polynomial time
    in  \cite{zanuttini03}.
  
    As for the $\ParityL$-hardness, let $B$ be such that $[B]=\CloneL_2$.
    Consider the $\ParityL$-complete problem to determine whether a system of linear equations $S$ over $GF(2)$ has a solution \cite{budaheme92}.
    Note that  $\ParityL$ is closed under complement, so deciding whether such a system has no solution is also  $\ParityL$-complete.
    Let $S=\{s_1,\ldots,s_m\}$ be such a system of linear equations over variables $\{x_1,\ldots,x_n\}$.
    Then, for all $1 \leq i \leq m$, the equation $s_i$ is of the form $ x_{i_1} + \cdots + x_{i_{n_i}} = c_i  \pmod 2$ with $c_i \in \{0,1\}$
    and $i_1,\ldots,i_{n_i} \in \{1,\ldots,n\}$.
    We map $S$ to a set of affine formulae $\Gamma=\{\varphi_1,\ldots,\varphi_m\}$ over variables $\{x_1,\ldots,x_n,q\}$ via
    \begin{center}
    \begin{tabular}{lll}
      $\varphi_i := \ x_{i_1} \xor \cdots \xor x_{i_{n_i}} \xor \true \ $	& if $c_i = 0$	& and	 \\
      $\varphi_i := \ x_{i_1} \xor \cdots \xor x_{i_{n_i}}$ 					 	    & if $c_i = 1$.
    \end{tabular}
    \end{center}
    Now define 
    \begin{align*}
      \Gamma' := \,
            & \{\varphi_i \xor q \mid \varphi_i \in \Gamma  \text{ is such that }\varphi_i(1,\ldots, 1)=0\} \\
      \cup \; & \{\varphi_i \mid \varphi_i \in \Gamma \text{ is such that }\varphi_i(1,\ldots, 1)=1\}.
    \end{align*}
    $\Gamma'$ is obviously satisfied by the assignment mapping all propositions to $\true$.
    It furthermore holds that $S$ has no solution if and only if $\Gamma' \land \neg q$ is unsatisfiable. Hence, we obtain that $S$ has no solution
    if and only if the propositional abduction problem $(\Gamma',\emptyset,q)$ has an explanation. 

    It remains to transform $\Gamma'$ into a set of $B$-formulae in logarithmic space. Since $[B]=\CloneL_2$, we have $x \xor y \xor z\in[B]$.
    We insert parentheses in every formula $\varphi$ of $\Gamma'$ in such a way that we get a ternary $\xor$-tree of logarithmic depth
    whose leaves are either a proposition or the constant 1. Then we replace every node $\xor$ by its equivalent $B$-formula. Thus we get a
    $(B\cup\{1\})$-formula of size  polynomial in the size of the original one. Lemma~\ref{lem:constants_sometimes_available} allows to conclude.
    \qed
  \end{proof}

    Observe that in the cases $[B] \subseteq \CloneL$, $[B] \subseteq \CloneE$, and $[B] \subseteq \CloneV$, 
    the abduction problem for $B$-formulae is self-reducible.
    Roughly speaking, this means that given an instance $\calP$ and a literal $l$,
    we can efficiently compute an instance $\calP'$ such that
    the question whether there exists an explanation $E$ with $l \in E$ 
    reduces to the question whether $\calP'$ admits solutions. 
    It is well-known that for self-reducible problems whose decision problem is in $\P$, the
    lexicographic first solution can be computed in $\FP$. 
    It is an easy exercise to extend this algorithm to enumerate all
    solutions in lexicographic order with polynomial delay and polynomial space. 
    Thus, if $[B] \subseteq \CloneL$ or $[B] \subseteq \CloneE$ or $[B] \subseteq \CloneV$, 
    the explanations of $\ABD(B,\MPQ)$ can be enumerated with polynomial delay and polynomial space 
    according to Proposition~\ref{prop:abd(PQ)-E-N-V-logspace} and~\ref{prop:abd(PQ)-L-P}.

  \begin{proposition}\label{prop:abd(PQ)-M-NP-c}
    Let $B$ be a finite set of Boolean functions such that $\CloneS_{00}  \subseteq [B] \subseteq \CloneM$ or $\CloneS_{10}  \subseteq [B]
    \subseteq \CloneM$ or $\CloneD_{2}  \subseteq [B] \subseteq \CloneM$. Then $\ABD(B,\MPQ)$ is $\NP$-complete.
  \end{proposition}
  
  \begin{proof}
    We first show that $\ABD(B,\MPQ)$ is efficiently verifiable. 
    Let $\calP=(\Gamma,A,q)$ be an $\ABD(B,\MPQ)$-instance and $E \subseteq \Lits{A}$ be a candidate for an explanation. 
    Define $\Gamma'$ as the set of formulae obtained from $\Gamma$ by replacing each occurrence of the proposition $x$ with $\false$ if $\neg x \in E$,
    and each occurrence of the proposition $x$ with $\true$ if $x \in E$.
    It holds that $E$ is a solution for $\calP$  if $\Gamma'$ is satisfiable and $\Gamma'[q/\false]$ is not. 
    These tests can be performed in polynomial time, because $\Gamma'$ is a set of monotonic formulae \cite{lew79}.
    Hence, $\ABD(B,\MPQ) \in \NP$.
    
    Next we give a reduction from the $\NP$-complete problem $\TWOINTHREESAT$, \emph{i.e.}, the problem to decide whether there exists an assignment that
    satisfies exactly two propositions in each clause of a given formula in conjunctive normal form with exactly three positive propositions
    per clause, see \cite{sch78}.
    Let $\varphi := \bigwedge_{i \in I} c_i$ with $c_i=x_{i1} \lor x_{i2} \lor x_{i3}$, $i \in I$, be the given formula.
    We map $\varphi$ to the following instance $\calP=(\Gamma,A,q)$.
    Let $q_i$, $i \in I$, be fresh, pairwise distinct propositions and let $A:= \Vars{\varphi} \cup \{q_i\mid i \in I\}$.
    The set $\Gamma$ is defined as
    \begin{align}
      \Gamma :=\,& \label{eq:q-abd-s01-1}
      \{c_i \mid i \in I \} \\
      \cup\; & \label{eq:q-abd-s01-2}
      \{ x_{i1} \lor x_{i2} \lor q_i, x_{i1} \lor x_{i3} \lor q_i, x_{i2} \lor x_{i3} \lor q_i \mid i \in I \} \\  
      \cup\; & \textstyle \label{eq:q-abd-s01-3}
      \{\bigvee_{i \in I} \bigwedge_{j=1}^3 x_{ij} \lor \bigvee_{i \in I} q_i \lor q\}.
    \end{align}
    We show that there is an assignment that sets to true exactly two propositions in each clause of $\varphi$ if and only if $\calP$ has a solution. 
    First, suppose that there exists an assignment $\sigma$ such that for all $i \in I$, there is a permutation $\pi_i$ of $\{1,2,3\}$ such that
    \ $\sigma(x_{i\pi_i(1)})=\false$ and $\sigma(x_{i\pi_i(2)})=\sigma(x_{i\pi_i(3)})=\true$. 
    Thus \eqref{eq:q-abd-s01-1} and \eqref{eq:q-abd-s01-2} are satisfied,
    and \eqref{eq:q-abd-s01-3} is equivalent to $\bigvee_{i \in I} q_i \lor q$.
    From this, it is readily observed that 
    $\{ \neg x \mid \sigma(x)=\false\} \cup \{\neg q_i \mid i \in I\}$ is a solution to $\calP$.

    Conversely, suppose that $\calP$ has an explanation $E$ that is w.l.o.g.\ full.
    Then $\Gamma \land E$ is satisfiable and $\Gamma \land E \models q$. 
    Let  $\sigma \colon \Vars{\Gamma}  \longrightarrow \{\false,\true\}$ be an assignment that satisfies $\Gamma \land E$. Then, for any $x \in A$,
    $\sigma(x)=\false$ if $\neg x \in E$, and $\sigma(x)=\true$ otherwise. Since $\Gamma \land E$ entails $q$
    and as the only occurrence of $q$ is in \eqref{eq:q-abd-s01-3}, 
    we obtain  that
    $\sigma$ sets to $\false$ each $q_i$ and at least one proposition in each clause of $\varphi$.
    Consequently, from \eqref{eq:q-abd-s01-2} it follows that $\sigma$ sets to $\true$  at least two propositions in each clause of $\varphi$.
    Therefore, $\sigma$ sets to $\true$ exactly two propositions in each clause of $\varphi$.
    
    It remains to show that $\calP$ can be transformed into an $\ABD(B,\MPQ)$-instance for all considered $B$.
    Observe that $\lor \in [B\cup\{1\}]$ and $[\CloneS_{00}\cup\{\false, \true\}] = [\CloneD_{2}\cup\{\false, \true\}] = [\CloneS_{10}\cup\{\false,
    \true\}] = \CloneM$. 
    Therefore due to Lemma~\ref{lem:constants_sometimes_available} it suffices to consider the case $[B]=\CloneM$.
    Using the associativity of $\lor$ rewrite \eqref{eq:q-abd-s01-3} as an $\lor$-tree of logarithmic depth and replace all the connectives in
    $\Gamma$ by their B-representation ($\lor,\land \in [B]$).  
    \qed
  \end{proof}

  \begin{proposition}\label{prop:abd(PQ)-SigP2-c}
    Let $B$ be a finite set of Boolean functions such that $\CloneS_{02} \subseteq [B]$ or $\CloneS_{12} \subseteq [B]$ or $\CloneD_{1} \subseteq [B]$.
    Then $\ABD(B,\MPQ)$ is $\SigPtwo$-complete.
  \end{proposition}
  
  \begin{proof}
    Membership in $\SigPtwo$ is easily seen to hold: given an instance $(\Gamma,A,q)$,
    guess an explanation $E$ and subsequently verify that $\Gamma \land E$ is satisfiable and $\Gamma \land E \land \neg q$ is not.
    
    Observe that $\lor \in [B\cup\{\true\}]$. By virtue of Lemma~\ref{lem:constants_sometimes_available} and the fact that
    $[\CloneS_{02} \cup \{\false,\true\}] = [\CloneS_{12} \cup \{\false,\true\}] = [\CloneD_{1} \cup \{\false,\true\}] = \CloneBF$,  
    it suffices to consider the case $[B] = \CloneBF$.
    In \cite{eiter-gottlob:1995} it has been shown that the propositional abduction problem
    remains $\SigPtwo$-complete when the knowledge base $\Gamma$ is a set of  clauses.
    From such an instance $(\Gamma, A, q)$ we build an instance of $\ABD(B,\MPQ)$ by rewriting first each clause as an $\lor$-tree of logarithmic depth
    and then replacing the occurring connectives $\lor$ and $\neg$ by their $B$-representation, thus concluding the proof.
    \qed
  \end{proof}

\subsection{Variants of $\ABD(B)$} \label{subsect:abduction-variants}
We now consider the symmetric abduction problem for different variants on the manifestations: clause, term and $B$-formula. 
Let us first make a remark on the cases where the manifestation is a (not necessarily positive) literal or a negative literal.

  \begin{remark}\label{rem:Q-PQ-NQ}
    $\ABD(B,\MQ)$ obeys the same classification as $\ABD(B,\MPQ)$ since all bounds, upper and lower, easily carry over.
    For $\ABD(B,\MNQ)$ the problem becomes trivial if $[B] \subseteq \CloneM$. For $[B] \subseteq \CloneL$ $\ABD(B,\MNQ)$ is solvable in
    polynomial time according to \cite{zanuttini03}.
    For the remaining clones (\emph{i.e.}, for $\CloneS_{02} \subseteq [B]$, $\CloneS_{12} \subseteq [B]$, and $\CloneD_1 \subseteq [B]$), 
    we can again easily adapt the proofs of $\ABD(B,\MPQ)$. 
    This way we obtain a dichotomous classification for $\ABD(B,\MNQ)$ into $\P$-complete and $\SigPtwo$-complete cases; thus skipping
    the intermediate $\NP$ level.
  \end{remark}

  For clauses, it is obvious that $\ABD(B,\MPQ) \leqlogm \ABD(B,\MPC)$. Therefore, all hardness results continue to hold for the $\ABD(B,\MPC)$.
  It is an easy exercise to prove that all algorithms that have been developed for a single query can be naturally extended to clauses.
  Therefore, the complexity classifications for the problems $\ABD(B,\MPC)$, $\ABD(B,\MC)$ and $\ABD(B,\MNC)$ are exactly the same as
  for $\ABD(B,\MPQ)$, $\ABD(B,\MQ)$ and $\ABD(B,\MNQ)$, respectively. 

  \begin{theorem}\label{thm:main_variant_PC}
    Let $B$
    be a finite set of Boolean functions. Then the symmetric abduction problem for propositional $B$-formulae with a positive clause manifestation, $\ABD(B,\MPC)$, is
    \begin{enumerate}
      \item $\SigPtwo$-complete if $\CloneS_{02}\subseteq [B]$ or $\CloneS_{12}\subseteq [B]$ or  $\CloneD_{1}\subseteq [B]$,
      \item $\NP$-complete if  $\CloneS_{00}\subseteq [B]\subseteq \CloneM$ or $\CloneS_{10}\subseteq [B]\subseteq \CloneM$ or $\CloneD_{2}
              \subseteq [B]  \subseteq \CloneM$,
      \item in $\P$ and $\ParityL$-hard if $\CloneL_{2}\subseteq [B]\subseteq \CloneL$, and
      \item in $\L$ in all other cases.
    \end{enumerate}
  \end{theorem}
  Notably, we
  will prove in the next section that allowing for terms as manifestations increases the complexity for the clones $\CloneV$ (from membership
  in $\L$ to $\NP$-completeness), while allowing $B$-formulae as manifestations makes the classification dichotomous again:
  all problems become either $\P$- or $\SigPtwo$-complete.

\subsubsection{The complexity of $\ABD(B,\MPT)$}\label{subsubsec:complexity_pt}
  
  \begin{proposition}\label{prop:S-ABD(PT)-V-NP-c}
    Let $B$ be a finite set of Boolean functions such that $\CloneV_2 \subseteq [B] \subseteq \CloneV$. Then $\ABD(B,\MPT)$ is $\NP$-complete.
  \end{proposition}
    
  \begin{proof}
    Let $B$ be a finite set of Boolean functions such that\ $\CloneV_2 \subseteq [B] \subseteq \CloneV$ and let $\calP=(\Gamma,A,t)$ be an
    instance of	$\ABD(B,\MPT)$. Hence, $\Gamma$ is a set of $B$-formulae and $t$ is a term, $t=\bigwedge_{i=1}^n l_i$. Observe that
    $E$ is a solution for $\calP$ if $\Gamma\land E$ is satisfiable and for every $i=1,\ldots, n$, $\Gamma\land E\land \neg l_i$ is not.
    Given a set $E\subseteq\Lits{A}$, these verifications, which  require substitution of symbols and evaluation of an $\lor$-formula,
    can be performed in polynomial time, thus proving membership in $\NP$.

    To prove $\NP$-hardness, we give a reduction from $\ThreeSAT$. 
    Let $\varphi$ be a 3-CNF-formula, 
    $\varphi := \bigwedge_{i \in I} c_i$.
    Let $x_{1}, \ldots, x_{n}$ enumerate the variables occurring in $\varphi$. 
    Let $x'_{1},\ldots, x'_{n}$ and $q_{1}, \ldots, q_{n}$ be fresh, pairwise distinct variables.
    We map $\varphi$ to $\calP=(\Gamma,A,t)$, where
    \begin{align*}
      \Gamma :=\,& \{ c_i[\neg x_1/x'_1,\ldots, \neg x_n/x'_n]\mid i\in I\}\\
      \cup\;& \{ x_i \lor x'_i,   x_i \lor q_i,  x'_i \lor q_i \mid 1 \leq i \leq n\}, \\
      A:=\,& \{x_1,\ldots,x_n,x'_1,\ldots,x'_n\}, \\
      t :=\,& q_1\land \cdots \land q_n.
    \end{align*}

    We show that $\varphi$ is satisfiable if and only if $\calP$ has a solution.
    First assume that $\varphi$ is satisfied by the assignment $\sigma\colon \{x_{1}, \ldots, x_{n}\} \longrightarrow \{\false,\true\}$. 
    Define $E:= \{ \neg x_i \mid \sigma(x_i)=\false \} \cup \{ \neg x'_i \mid \sigma(x_i)=\true \}$ and 
    $\hat\sigma$ as the extension of $\sigma$ mapping $\hat\sigma(x'_i)=\neg \sigma(x_i)$ and $\hat\sigma(q_i)=\true$ for all $1 \leq i \leq n$. 
    Obviously, $\hat\sigma \models \Gamma \land E$. 
    Furthermore, $\Gamma \land E \models q_i$ for all $1 \leq i \leq n$, 
    because any satisfying assignment of $\Gamma \wedge E$ sets to $\false$ either $x_i$ or $x_i'$ and thus $\{  x_i \lor q_i,  x'_i \lor q_i\}
    \models q_i$. Hence $E$ is an explanation for $\calP$.
    
    Conversely, suppose that $\calP$ has a full explanation $E$. The facts that  $\Gamma\land E\models q_1\land \cdots \land q_n $  and that each $q_i$
    occurs only in the clauses  $ x_i \lor q_i,  x'_i \lor q_i$ enforce that, for every $i$, $E$  contains $\neg x_i$ or $\neg x'_i$. Because of the
    clause $x_i \lor x'_i$, it cannot contain both. Therefore in $E$ the value of $x'_i$ is determined by the value of $x_i$ and is its dual. From this
    it is easy to conclude that the assignment $\sigma\colon \{x_{1}, \ldots, x_{n}\} \longrightarrow \{\false,\true\}$ defined by $\sigma(x_i)=\false$ if $\neg
    x_i\in E$, and $\true$ otherwise, satisfies $\varphi$.
    Finally $\calP$ can be transformed into an $\ABD(B,\MPT)$-instance, because every formula in $\Gamma$ is the disjunction of at most three
    variables and $\lor\in [B]$.
    \qed
  \end{proof}

  \begin{theorem}\label{thm:main_variant_PT}
    Let $B$
    be a finite set of Boolean functions. Then the symmetric abduction problem for propositional $B$-formulae with a positive term manifestation, $\ABD(B,\MPT)$, is
    \begin{enumerate}
      \item $\SigPtwo$-complete if $\CloneS_{02}\subseteq [B]$ or $\CloneS_{12}\subseteq [B]$ or  $\CloneD_{1}\subseteq [B]$,
      \item $\NP$-complete if  $\CloneV_2\subseteq [B]\subseteq \CloneM$ or $\CloneS_{10}\subseteq [B]\subseteq \CloneM$ or $\CloneD_{2}\subseteq [B]
              \subseteq \CloneM$,
      \item in $\P$ and $\ParityL$-hard if $\CloneL_{2}\subseteq [B]\subseteq \CloneL$, and
      \item in $\L$ in all other cases.
    \end{enumerate}
  \end{theorem}

  \begin{proof}
    \begin{enumerate}
      \item The $\SigPtwo$-hardness follows directly from Proposition~\ref{prop:abd(PQ)-SigP2-c}.
      \item For the clones $\CloneV_2 \subseteq [B] \subseteq \CloneV$, see Proposition~\ref{prop:S-ABD(PT)-V-NP-c}. In all other clones, 
            the $\NP$-hardness follows from a straightforward generalization of the proof of Proposition~\ref{prop:abd(PQ)-M-NP-c}.
      \item Membership in $\P$ follows directly from \cite[Theorem~67]{noza08}, the $\ParityL$-hardness from Proposition~\ref{prop:abd(PQ)-L-P}.
      \item Analogous to Proposition~\ref{prop:abd(PQ)-E-N-V-logspace}. 
      \qed
    \end{enumerate}
  \end{proof}

  \begin{remark}\label{rem:T-PT-NT}
    All upper and lower bounds for $\ABD(B,\MPT)$ easily carry over to $\ABD(B,\MT)$.
    It is also easily seen that $\ABD(B,\MNT)$ is classified exactly as $\ABD(B,\MNQ)$, see Remark~\ref{rem:Q-PQ-NQ}.
  \end{remark}

\subsubsection{The complexity of $\ABD(B,\MF)$}\label{subsubsec:complexity_lb}

  \begin{proposition}\label{prop:abd(MF)-SigP2-c}
    Let $B$ be a finite set of Boolean functions such that $\CloneS_{00} \subseteq [B]$ or $\CloneS_{10} \subseteq [B]$ or $\CloneD_2 \subseteq [B]$.
    Then $\ABD(B,\MF)$ is $\SigPtwo$-complete.
  \end{proposition}

  \begin{proof}
    We prove $\SigPtwo$-hardness by giving a reduction from the $\SigPtwo$-hard problem $\QSAT_2$ \cite{wra77}. Let an instance of $\QSAT_2$ be given
    by a closed formula $\chi:=\exists x_1 \cdots \exists x_n \forall y_1 \cdots \forall y_m \varphi$ with $\varphi$ being a 3-DNF-formula.
    First observe that $\exists x_1 \cdots \exists x_n \forall y_1 \cdots \forall y_m \varphi$ is true if and only if there exists a consistent set
    $X \subseteq \Lits{\{x_1,\ldots,x_n\}}$ such that $X \cap \{x_i,\neg x_i\} \neq \emptyset$, for all $1 \leq i \leq n$, and $\neg X \lor \varphi$
    is (universally) valid (or equivalently $\neg\varphi \land X$ is unsatisfiable). 

    Denote by $\overline{\varphi}$ the negation normal form of $\neg \varphi$ and
    let $\overline{\varphi}'$   be obtained from $\overline{\varphi}$ by replacing 
    all occurrences of $\neg x_i$ with a fresh proposition $x_i'$, $1\leq i \leq n$, and
    all occurrences of $\neg y_i$ with a fresh proposition $y_i'$, $1\leq i \leq m$.
    That is, $\overline{\varphi}' \equiv \overline{\varphi}[\neg x_1/x'_1,\ldots,\neg x_n/x'_n, \neg y_1/y'_1,\ldots,\neg y_m/y'_m]$. Thus
    $\overline{\varphi}'=\bigwedge _{i\in I} c'_i$, where every $c'_i$ is a disjunction of  three propositions. To $\chi$ we associate the
    propositional abduction problem $  \calP=(\Gamma,A,\psi)$ defined as follows:
    \begin{align*}
    \Gamma:=\,& \{c'_i \lor q\mid   i \in I\} \\ 
      \cup\;& \{ x_i \lor x'_i \mid 1 \leq i \leq n\} \cup \{ y_i \lor y'_i \mid 1 \leq i \leq m\} \\
      \cup\;& \{f_i \lor x_i, t_i \lor x'_i, f_i \lor t_i \mid 1 \leq i \leq n \}, \\
        A :=\,& \{t_i,f_i \mid 1 \leq i \leq n \},\\
      \psi :=\,&\textstyle q \lor \bigvee_{1 \leq i \leq n} (x_i \land x'_i) \lor \bigvee_{1 \leq i \leq m} (y_i \land y'_i).
    \end{align*}
  
    Suppose that $\chi$ is true.
    Then there exists an assignment $\sigma\colon\{x_1,\ldots,x_n\} \longrightarrow \{\false,\true\}$ such that no extension
    $\sigma'\colon\{x_1,\ldots,x_n\} \cup \{y_1,\ldots,y_m\} \longrightarrow \{\false,\true\}$ of $\sigma$ satisfies $\neg \varphi$.
    Define $X$ as the set of literals over $\{x_1,\ldots,x_n\}$ set to $\true$ by $\sigma$.
    Defining $E:= \{ \neg f_i, t_i \mid x_i \in X \} \cup \{ \neg t_i, f_i \mid \neg x_i \in X \}$, we obtain with abuse of notation
    \begin{align*}
    \Gamma \land E \land \neg \psi \;	 \equiv \; 	&	\textstyle \bigwedge_{i \in I} c'_i \land \bigwedge_{1 \leq i \leq n} (x_i \xor x'_i)\land
                                                        \bigwedge_{1 \leq i \leq m}	(y_i \xor y'_i)  \land & \\ 
                                                  & \textstyle \bigwedge_{1 \leq i \leq n, \sigma(x_i)=1} x_i \land \bigwedge_{1 \leq i \leq
                                                        n, \sigma(x_i)=0} x'_i & \\
                                      \equiv \;	& \neg\varphi \land X, &
    \end{align*}	
    
    which is unsatisfiable by assumption.
    As $\Gamma \land E$ is satisfied by any assignment setting in addition all $x_i, x'_i$, $1 \leq i \leq n$, and all $y_j, y'_j$,
    $1 \leq i \leq m$, to $\true$, we have proved that $E$ is an explanation for $\calP$.

    Conversely, suppose that $\calP$ has an explanation $E$. 
    Assume w.l.o.g.\ that $E$ is full. 
    Due to the clause $(f_i\lor t_i)$ in $\Gamma$,
    we also may assume that $|E \cap \{\neg t_i, \neg f_i\}| \leq 1$ for all $1 \leq i \leq n$. 

    Setting $X := \{x_i \mid \neg f_i \in E\} \cup \{\neg x_i \mid \neg t_i \in E\}$ we now obtain
    $\bigwedge_{1 \leq i \leq n} \big((f_i \lor x_i) \land (t_i \lor \neg x_i)\land (f_i\lor t_i)\big) \land E\equiv X$ and 
    $\Gamma \land E \land \neg \psi \equiv \neg\varphi \land X$ as above.
    Hence, $\neg\varphi \land X$ is unsatisfiable, which implies the existence of an assignment 
    $\sigma\colon\{x_1,\ldots,x_n\} \longrightarrow \{\false,\true\}$ such that no extension 
    $\sigma'\colon\{x_1,\ldots,x_n\} \cup \{y_1,\ldots,y_m\} \longrightarrow \{\false,\true\}$ of $\sigma$ satisfies $\neg \varphi$.
    Therefore, we have proved that $\chi$ is true if and only if $\calP$ has an explanation. 

    It remains to show that $\calP$ can be transformed into an $\ABD(B,\MF)$-instance for any relevant $B$.
    Since $[\CloneS_{00} \cup \{\true\}] = \CloneS_{01}$, $[\CloneS_{10} \cup \{\true\}] = \CloneM_1$, $[\CloneD_2 \cup \{\true\}] = \CloneS^2_{01}$ 
    and $\CloneS_{01} \subseteq \CloneS^2_{01} \subseteq\CloneM_1$, it suffices to consider
    the case $[B] = \CloneS_{01}$ by Lemma~\ref{lem:constants_sometimes_available}. 
    Observe that $x \lor (y \land z) \in [B]$. The transformation can hence be done in polynomial time by
    local replacements, rewriting $\psi$ as $\bigvee_{1 \leq i \leq n} q \lor (x_i \land x'_i) \lor \bigvee_{1 \leq i \leq m} q \lor (y_i \land y'_i)$
    and using the associativity of $\lor$. 
    \qed
  \end{proof}

  \begin{theorem}\label{thm:main_variant_F}
    Let $B$ be a finite set of Boolean functions. 
    Then the symmetric abduction problem for propositional $B$-formulae with a $B$-formula manifestation, $\ABD(B,\MF)$, is
    \begin{enumerate}
      \item $\SigPtwo$-complete if $\CloneS_{00}\subseteq [B]$ or $\CloneS_{10}\subseteq [B]$ or  $\CloneD_{2}\subseteq [B]$,
      \item in $\P$ and $\ParityL$-hard if $\CloneL_{2}\subseteq [B]\subseteq \CloneL$, and
      \item in $\L$ in all other cases.
    \end{enumerate}
  \end{theorem}

  \begin{proof}
    \begin{enumerate}
      \item See Proposition~\ref{prop:abd(MF)-SigP2-c}.
      \item Membership in $\P$ follows directly from \cite{zanuttini03}, 
      the $\ParityL$-hardness follows from Proposition~\ref{prop:abd(PQ)-L-P}.
      \item Analogous to Proposition~\ref{prop:abd(PQ)-E-N-V-logspace}. 
      \qed
    \end{enumerate}
  \end{proof}

  Observe that there are no sets $B$ of Boolean functions for which $\ABD(B,\MF)$ is $\NP$-complete.

\section{The complexity of $\ABD[\HP](B)$}\label{sec:positive_abduction}

  We will now study the complexity of positive abduction, in which an explanation consists of a set of positive literals. 
  The results of this section are summarized in Figure~\ref{fig:positive_abduction_complexity}. 
  To begin with, note that for monotonic or $\true$-reproducing sets of formulae, 
  deciding the existence of a positive explanation reduces to a testing whether $A$ is one.

  \begin{lemma}\label{lem:p-abd-useful}
    For $[B]\subseteq \CloneR_{1}$ or $[B]\subseteq \CloneM$, an instance $(\Gamma, A, \varphi)$ of $\ABD[\HP](B,\manif)$ has
    solutions if and only if $A$ is a solution.
  \end{lemma}
  \begin{proof}
    Let $E \subseteq A$ be an arbitrary explanation for the given instance $(\Gamma,A,\varphi)$.
    One easily verifies that 
    $\Gamma \land A$ remains satisfiable if $\Gamma \land E$ was, and
    $\Gamma \land A \land \neg \varphi$ remains unsatisfiable if $\Gamma \land E \land \neg \varphi$ was.
    Conversely, if $A$ is not an explanation, no proper subset $E \subseteq A$ can be an explanation either.
    \qed
  \end{proof}
  
  As a consequence we will see that some of the formerly $\NP$-complete cases become tractable
  and that some of the formerly $\SigPtwo$-complete cases become $\coNP$-complete.

  
\subsection{The complexity of $\ABD[\HP](B,\MPQ)$}\label{subsec:complexity_pos_pq}

  \begin{proposition}\label{prop:P-abd(PQ)-M-L}
    Let $[B]\subseteq\CloneM$. Then $\ABD[\HP](B, \MPQ)\in \L$.
  \end{proposition}
  \begin{proof}
    According to Lemma~\ref{lem:p-abd-useful} it suffices to test if $A$ is a solution, that is, to test if $\Gamma\land A\neg q$ or equivalently
    $\Gamma\land A[q/0]$ is unsatisfiable. This can be done in  logarithmic space, since $\Gamma\land A$ is a
    monotonic formula \cite{lew79}.
    \qed
  \end{proof}

  \begin{proposition}\label{prop:P-abd(PQ)-R1-coNP-c}
    Let $\CloneS_{02} \subseteq [B] \subseteq \CloneR_1$ or $\CloneS_{12} \subseteq [B] \subseteq \CloneR_1$ or
    $\CloneD_1 \subseteq [B] \subseteq \CloneR_1$.
    Then $\ABD[\HP](B, \MPQ)$ is $\coNP$-complete.
  \end{proposition}
  \begin{proof}
    According to Lemma~\ref{lem:p-abd-useful} it suffices to test whether $A$ is a solution.
    Since $\Gamma \wedge A$ is always satisfiable, only the task of testing whether 
    $\Gamma \wedge A \wedge \neg q$ is unsatisfiable remains. And this can be done in $\coNP$.

    Since $[\CloneD_1 \cup \{\true\}] = [\CloneS_{12} \cup \{\true\}] = \CloneR_1$, $[\CloneS_{02} \cup \{\true\}] = \CloneS_0$ and
    $\CloneS_0 \subseteq \CloneR_1$, it suffices to show hardness for the case $\CloneS_0\subseteq[B]$ by Lemma~\ref{lem:constants_sometimes_available}. 
    To show $\coNP$-hardness we give a reduction from $\overline{\ThreeSAT}$. Let $\varphi = \bigwedge_{i\in I}c_i$
    be a 3-CNF-formula, \emph{i.e.}, each clause $c_i$ consists of the disjunction of exactly three literals. Since $[\{\limplies,\false\}]=\CloneBF$,
    each clause $c_i$ has a representation as a $\{\limplies,\false\}$-formula which we indicate by $c'_i$. 
    Let $q$ be a fresh proposition. We map $\varphi$ to $(\Gamma, \emptyset, q)$, where we define $\Gamma = \bigcup_{i\in I}c'_i[\false / q]$.
    Note that $\Gamma$ is a set of $\CloneS_0$-formulae of polynomial size and $\true$-reproducing.
    Let $\varphi$ be unsatisfiable. Then $\Gamma$ is satisfied by the assignment setting to $\true$ all propositions and 
    $\Gamma \wedge \neg q$ is unsatisfiable, because it is equivalent to $\varphi \wedge \neg q$. 
    Summing up, $\emptyset$ is an	explanation for $(\Gamma, \emptyset, q)$.
    Conversely, let $\varphi$ be satisfiable. 
    This implies that $\Gamma \wedge \neg q$ is satisfiable and thus $(\Gamma, \emptyset, q)$ has no explanations.
    
    It remains to transform $(\Gamma, \emptyset, q)$ into a $\ABD[\HP](B, \MPQ)$-instance for any relevant $B$.
    As $\mathord{\limplies}\in\CloneS_0\subseteq[B]$, this is done by replacing in $\Gamma$ every occurrence of 
    $\limplies$ by its $B$-representation.
    \qed
  \end{proof}

  \begin{theorem}\label{thm:P-abd(PQ)_main}
    Let $B$ be a finite set of Boolean functions. Then the positive abduction problem for propositional $B$-formulae  with a positive literal manifestation, $\ABD[\HP](B,\MPQ)$, is
    \begin{enumerate}
      \item $\SigPtwo$-complete if $\CloneD\subseteq [B]$ or $\CloneS_1\subseteq [B]$,
      \item $\coNP$-complete if $\CloneS_{02}\subseteq [B]\subseteq \CloneR_1$ or $\CloneS_{12}\subseteq [B]\subseteq \CloneR_1$
            or  $\CloneD_{1}\subseteq [B]\subseteq \CloneR_1$,
      \item in $\P$ and $\ParityL$-hard if $\CloneL_{2}\subseteq [B]\subseteq \CloneL$,
      \item in $\L$ in all other cases.
    \end{enumerate}
  \end{theorem}

  \begin{proof}
    \begin{enumerate}
      \item In \cite{noza08}, Nordh and Zanuttini prove that the abduction problem 
            in which the knowledge base is a set of clauses remains $\SigPtwo$-hard even 
            if explanations are required to comprise positive literals only. A reduction 
            from this problem can be done analogously to the one in Proposition~\ref{prop:abd(PQ)-SigP2-c}.
      \item See Proposition~\ref{prop:P-abd(PQ)-R1-coNP-c}.
      \item Membership in $\P$ follows from \cite[Theorem~66]{noza08}. For the $\ParityL$-hardness the same reduction as in
            Proposition~\ref{prop:abd(PQ)-L-P} works.
      \item See Proposition~\ref{prop:P-abd(PQ)-M-L}. 
      \qed
    \end{enumerate}
  \end{proof}

\subsection{Variants of $\ABD[\HP](B)$}\label{subsec:positive_variants}

Having examined the complexity of positive abduction for positive literal manifestations, 
we will now consider positive abduction for manifestation that are restricted to be 
respectively a clause, a term, or a $B$-formula. But first let us make a remark on the 
complexity of positive abduction when the manifestation is a (not necessarily positive) literal or a negative literal.

  \begin{remark}\label{rem:p-abd-Q-PQ-NQ}
    Again all upper and lower bounds for $\ABD[\HP](B,\MPQ)$ easily carry over to $\ABD[\HP](B,\MQ)$.
    For $\ABD[\HP](B,\MNQ)$ the problem becomes trivial if $[B] \subseteq \CloneR_1$ (Lemma~\ref{lem:p-abd-useful}).
    For $\CloneL_0 \subseteq [B] \subseteq \CloneL$ and $\CloneL_3 \subseteq [B] \subseteq \CloneL$, we obtain membership in $\P$
    from \cite[Theorem~66]{noza08}. For all remaining cases (\emph{i.e.}, for $\CloneD\subseteq [B]$ and $\CloneS_1\subseteq [B]$), we obtain
    $\SigPtwo$-completeness from an easy adaption of the first part in the proof of Proposition~\ref{thm:P-abd(PQ)_main}.
  \end{remark}
  
  Analogously to the symmetric case the algorithms can be extended to clauses. Thus, $\ABD[\HP](B, \MPC)$ is classified as $\ABD[\HP](B, \MPQ)$. Similarly the classifications for $\ABD[\HP](B, \MC)$ and $\ABD[\HP](B, \MNC)$ are the same as classifications for $\ABD[\HP](B, \MQ)$ and $\ABD[\HP](B, \MNQ)$.

  \begin{theorem}\label{thm:P-abd(PC)_main}
    Let $B$ be a finite set of Boolean functions. Then the positive abduction problem for propositional $B$-formulae  with a positive clause manifestation, $\ABD[\HP](B,\MPC)$, is
    \begin{enumerate}
      \item $\SigPtwo$-complete if $\CloneD\subseteq [B]$ or $\CloneS_1\subseteq [B]$,
      \item $\coNP$-complete if $\CloneS_{02}\subseteq [B]\subseteq \CloneR_1$ or $\CloneS_{12}\subseteq [B]\subseteq \CloneR_1$
            or  $\CloneD_{1}\subseteq [B]\subseteq \CloneR_1$,
      \item in $\P$ and $\ParityL$-hard if $\CloneL_{2}\subseteq [B]\subseteq \CloneL$,
      \item in $\L$ in all other cases.
    \end{enumerate}
  \end{theorem}
  
\subsubsection{The Complexity of $\ABD[\HP](B, \MPT)$}\label{subsubsec:complexity_pos_pt}

  The classification for positive terms is identical to the one for a single positive literal, except for the affine clones $\CloneL_0$, $\CloneL_3$, and $\CloneL$. For these, the complexity of $\ABD[\HP](B, \MPT)$ jump from membership in $\P$ to $\NP$-completeness.
  
  \begin{proposition}\label{prop:P-abd(PT)-L0-L3-NP-c}
    Let $\CloneL_0\subseteq [B]\subseteq \CloneL$ or $\CloneL_3\subseteq [B]\subseteq \CloneL$. Then $\ABD[\HP](B,\MPT)$ is $\NP$-complete.
  \end{proposition}
  \begin{proof}
    Let $(\Gamma, A, t)$ with $t=\bigwedge_{i\in I} x_i$ be an instance of $\ABD[\HP](B,\MPT)$. To check whether a given $E\subseteq A$ is an
    explanation, we have to test the
    satisfiability of $\Gamma \wedge E$ and the unsatisfiability of $\Gamma \wedge E \wedge \neg x_i$ for every $i\in I$. These tasks are equivalent
    to solving systems of linear equations, which can be done in polynomial time.
    The hardness follows directly from \cite[Theorem~70]{noza08}.
    \qed
  \end{proof}
  
  \begin{theorem}\label{thm:P-abd(PT)_main}
    Let $B$ be a finite set of Boolean functions. Then the positive abduction problem for propositional $B$-formulae  with a positive term manifestation, $\ABD[\HP](B,\MPT)$, is
    \begin{enumerate}
      \item $\SigPtwo$-complete if $\CloneD\subseteq [B]$ or $\CloneS_1\subseteq [B]$,
      \item $\coNP$-complete if $\CloneS_{02}\subseteq [B]\subseteq \CloneR_1$ or $\CloneS_{12}\subseteq [B]\subseteq \CloneR_1$
            or  $\CloneD_{1}\subseteq [B]\subseteq \CloneR_1$,
      \item $\NP$-complete if $[B] \in \{\CloneL,\CloneL_0,\CloneL_3\}$,
      \item in $\P$ and $\ParityL$-hard if $[B] \in \{\CloneL_1,\CloneL_2\}$,
      \item in $\L$ in all other cases.
    \end{enumerate}
  \end{theorem}

  \begin{proof}
    \begin{enumerate}
      \item Follows from the first item of Proposition~\ref{thm:P-abd(PQ)_main}.
      \item Both membership and hardness follow from Proposition~\ref{prop:P-abd(PQ)-R1-coNP-c}.
      \item See Proposition~\ref{prop:P-abd(PT)-L0-L3-NP-c}.
      \item For the $\ParityL$-hardness the same reduction as in Proposition~\ref{prop:abd(PQ)-L-P} works.
            Since $\CloneL_1 \subseteq \CloneR_1$, according to Lemma~\ref{lem:p-abd-useful}, it suffices to check whether $A$ is a solution.
            This task reduces to solve systems of linear equations which is in $\P$.
      \item Analogously to Proposition~\ref{prop:P-abd(PQ)-M-L}.
      \qed
    \end{enumerate}
  \end{proof}

  \begin{remark}\label{rem:p-abd-T-PT-NT}
    Again all upper and lower bounds for $\ABD[\HP](B,\MPT)$ carry over to $\ABD[\HP](B,\MT)$.
    For $\ABD[\HP](B,\MNT)$ the problem becomes trivial if $[B] \subseteq \CloneR_1$ (Lemma~\ref{lem:p-abd-useful}).
    For $[B] \in \{\CloneL,\CloneL_0,\CloneL_3\}$, we obtain $\NP$-completeness with
    the hardness being obtained from an easy reduction from $\ABD[\HP](B,\MPT)$: 
    as we have $x \xor y \in [B \cup \{\true\}]$, we can simply transform the given positive term into a negative one.
    For the remaining cases (\emph{i.e.}, $\CloneD\subseteq [B]$ and $\CloneS_1\subseteq [B]$), we obtain 
    $\SigPtwo$-completeness from an adaption of the first part in the proof of Proposition~\ref{thm:P-abd(PQ)_main}.
  \end{remark}

\subsubsection{The Complexity of $\ABD[\HP](B, \MF)$}\label{subsubsec:complexity_pos_lb}

  The complexity of $\ABD[\HP](B,\MF)$ differs from the complexity of $\ABD[\HP](B,\MPQ)$ 
  for the clones either (a) above $\CloneE$ or $\CloneV$ and below $\CloneM$ or (b) above $\CloneL_0$ or $\CloneL_3$ and below $\CloneL$.
  For the former the complexity increases to $\coNP$-completeness, 
  whereas for the latter we obtain membership in $\NP$ and hardness for $\ParityL$;
  the exact complexity of positive abduction when both the knowledge base and the manifestation
  are represented by non-$\true$-reproducing affine formulae remains an open problem.
  
  \begin{proposition}\label{prop:P-abd(MF)-L-NP}
    Let $B$ be a finite set of Boolean functions such that
    $\CloneL_0\subseteq [B]\subseteq \CloneL$ or $\CloneL_3\subseteq [B]\subseteq \CloneL$.
    Then $\ABD[\HP](B,\MF) \in \NP$.
  \end{proposition}
  \begin{proof}
    Let $E\subseteq A$ be a potential solution. The test for satisfiability of $\Gamma \wedge E$ and the test for unsatisfiability of
    $\Gamma \wedge E \wedge \neg \varphi$ are equivalent to solving two systems of linear equations, which can be done in polynomial time.
    \qed
  \end{proof}

  \begin{proposition}\label{prop:P-abd(MF)-M-coNP-c}
    Let $B$ be a finite set of Boolean functions such that
    $\CloneS_{00} \subseteq [B] \subseteq \CloneM$ or $\CloneS_{10} \subseteq [B] \subseteq \CloneM$ or $\CloneD_2 \subseteq [B] \subseteq \CloneM$.
    Then $\ABD[\HP](B,\MF)$ is $\coNP$-complete.
  \end{proposition}

  \begin{proof}
    We will first prove $\co\NP$-membership. 
    According to Lemma~\ref{lem:p-abd-useful}, it suffices to test whether $A$ is a solution. 
    This can be done by first verifying that $\Gamma \land A$ is satisfiable, and afterwards 
    verifying that $\Gamma \wedge A \wedge \neg \varphi$ is unsatisfiable. 
    As $\Gamma$ is a set of monotonic formulae, deciding the satisfiability of $\Gamma \wedge A$ 
    can be done in logarithmic space; and deciding whether $\Gamma \wedge A \wedge \neg \varphi$ is unsatisfiable is in $\coNP$.
    
    To establish $\coNP$-hardness, we give a reduction from the $\coNP$-hard problem to decide whether a given 3-DNF-formulae $\varphi$ is a tautology.
    Let $\Vars{\varphi} = \{x_1, \dots x_n\}$. Denote by $\overline{\varphi}$ the negation normal form of $\neg \varphi$ and
    let $\overline{\varphi}'$ be obtained from $\overline{\varphi}$ by replacing 
    all occurrences of $\neg x_i$ with a fresh proposition $x_i'$, $1\leq i \leq n$.
    That is, $\overline{\varphi}' \equiv \overline{\varphi}[\neg x_1/x'_1,\ldots,\neg x_n/x'_n]$. Thus
    $\overline{\varphi}'=\bigwedge _{i\in I} c'_i$ where every $c'_i$ is a disjunction of  three propositions. To $\varphi$ we associate the
    propositional abduction problem $  \calP=(\Gamma,\emptyset,\psi)$ defined as follows:
    \begin{align*}
    \Gamma:=\,& \{c'_i \lor q\mid   i \in I\}\; \cup\;  \{ x_i \lor x'_i \mid 1 \leq i \leq n\}, \\
      \psi :=\,&\textstyle q \lor \bigvee_{1 \leq i \leq n} (x_i \land x'_i).
    \end{align*}

    \noindent
    Observe that 
    \begin{align}
      \label{eq:conp} \Gamma \land \neg \psi \;	 \equiv \;	\neg\varphi \wedge \neg q \wedge \bigwedge_{i=1}^{n}(x_i \xor x'_i).
    \end{align}
    Suppose that $\varphi$ is a tautology, \emph{i.e.}, $\neg \varphi$ is unsatisfiable.  From \eqref{eq:conp} it follows that 
    $\Gamma \wedge \neg \psi$ is unsatisfiable. As $\Gamma$ is satisfiable, $\emptyset$ is a solution for $\calP$.
    
    Suppose conversely that $\emptyset$ is a solution for $\calP$. 
    Then $\Gamma \land \neg \psi \equiv \neg\varphi \wedge \neg q \wedge \bigwedge_{i=1}^{n}(x_i \xor x'_i)$ is unsatisfiable. 
    Since  $q$ and the $x'_i$ do not occur in $\varphi$, we obtain the unsatisfiability of $\neg \varphi$. Hence, $\varphi$ is a tautology.
    
    The transformation of $\calP$ into an $\ABD[\HP](B,\MF)$-instance for any relevant $B$ can be done in exactly the same way as in the proof of
    Proposition~\ref{prop:abd(MF)-SigP2-c}.
    \qed
  \end{proof}

  \begin{theorem}\label{thm:P-abd(F)_main}
    Let $B$ be a finite set of Boolean functions. Then the positive abduction problem for propositional $B$-formulae  with a $B$-formula manifestation, $\ABD[\HP](B,\MF)$, is
    \begin{enumerate}
      \item $\SigPtwo$-complete if $\CloneD\subseteq [B]$ or $\CloneS_1\subseteq [B]$,
      \item $\coNP$-complete if
        $\CloneS_{02}\subseteq [B]\subseteq \CloneR_1$ or $\CloneS_{12}\subseteq [B]\subseteq \CloneR_1$ or $\CloneD_{1}\subseteq [B]\subseteq \CloneR_1$
        or
        $\CloneS_{00} \subseteq [B] \subseteq \CloneM$ or $\CloneS_{10} \subseteq [B] \subseteq \CloneM$ or $\CloneD_2 \subseteq [B] \subseteq \CloneM$ 
      \item in $\NP$ and $\ParityL$-hard if $[B] \in \{\CloneL,\CloneL_0,\CloneL_3\}$,
      \item in $\P$ and $\ParityL$-hard if $[B] \in \{\CloneL_1,\CloneL_2\}$,
      \item in $\L$ in all other cases.
    \end{enumerate}
  \end{theorem}

  \begin{proof}
    \begin{enumerate}
      \item Follows from the first item of Proposition~\ref{thm:P-abd(PQ)_main}.
      \item See Proposition~\ref{prop:P-abd(MF)-M-coNP-c} and~\ref{prop:P-abd(PQ)-R1-coNP-c}.
      \item For membership in $\NP$, see Proposition~\ref{prop:P-abd(MF)-L-NP}.
            The $\ParityL$-hardness, on the other hand, is established using the same reduction as in the proof of Proposition~\ref{prop:abd(PQ)-L-P}.
      \item For membership in $\P$, see the fourth item of Proposition~\ref{thm:P-abd(PT)_main}. 
            The $\ParityL$-hardness follows from Proposition~\ref{prop:abd(PQ)-L-P} as above.
      \item For $[B] \subseteq \CloneV$, a $B$-formula is a positive clause. Thus the result follows from Theorem~\ref{thm:P-abd(PC)_main}.
            For $[B] \subseteq \CloneN$ and $[B] \subseteq \CloneE$, see Proposition~\ref{prop:abd(PQ)-E-N-V-logspace}.
      \qed
    \end{enumerate}
  \end{proof}

\section{Overview of   Results}\label{sec:overview}
The following two tables give an overview of the results for the studied symmetric and positive abduction problems. 
The small numbers on the right side in the table cells refer to the corresponding theorem/proposition/remark. 
The number is omitted for trivial results.

\begin{table}[H]
\setlength{\tabcolsep}{3pt}
\renewcommand{\arraystretch}{1.2}
  \begin{tabularx}{\linewidth}{X|cc|cc|cc|cc|cc|cc}
      Manifestation 																					&
      \multicolumn{2}{c|}{$\CloneE_*$} 												&
      \multicolumn{2}{c|}{$\CloneN_*$} 												&
      \multicolumn{2}{c|}{$\CloneV_*$}												&
      \multicolumn{2}{c|}{$\CloneL_*$} 												&
      \multicolumn{2}{c|}{$\CloneD_2,\CloneS_{*0}\subseteq$}	&
      \multicolumn{2}{c}{$\CloneD_1,\CloneS_{*2}\subseteq$}	\\
                                                  &
      \multicolumn{2}{c|}{} 											&
      \multicolumn{2}{c|}{} 											&
      \multicolumn{2}{c|}{}												&
      \multicolumn{2}{c|}{} 											&
      \multicolumn{2}{c|}{$[B]\subseteq\CloneM$} 	&
      \multicolumn{2}{c}{$[B]\subseteq\CloneBF$} \\
    \hline
      $\MNQ,\MNC,\MNT$	&
      $\in \sL$					& {\mini\ref{rem:Q-PQ-NQ}} &
      $\in \sL$ 					& {\mini\ref{rem:Q-PQ-NQ}} &
      $\in \sL$ 					& {\mini\ref{rem:Q-PQ-NQ}} &
      $\in \P$ 					& {\mini\ref{prop:abd(PQ)-L-P}} &
      $\in \sL$ 					& {\mini\ref{rem:Q-PQ-NQ}} &
      $\SigPtwo$-c. 		& {\mini\ref{rem:Q-PQ-NQ}} \\
      $\MPQ,\MPC,\MQ,\MC$	&
      $\in \sL$					& {\mini\ref{prop:abd(PQ)-E-N-V-logspace}} &
      $\in \sL$					& {\mini\ref{prop:abd(PQ)-E-N-V-logspace}} &
      $\in \sL$					& {\mini\ref{prop:abd(PQ)-E-N-V-logspace}} &
      $\in \P$					& {\mini\ref{prop:abd(PQ)-L-P}} &
      $\NP$-c.					& {\mini\ref{prop:abd(PQ)-M-NP-c}} &
      $\SigPtwo$-c.			& {\mini\ref{prop:abd(PQ)-SigP2-c}} \\
      $\MPT,\MT$				&
      $\in \sL$ 				& {\mini\ref{thm:main_variant_PT}} &
      $\in \sL$ 				& {\mini\ref{thm:main_variant_PT}} &
      $\NP$-c. 					& {\mini\ref{prop:S-ABD(PT)-V-NP-c}} &
      $\in \P$ 					& {\mini\ref{thm:main_variant_PT}} &
      $\NP$-c. 					& {\mini\ref{thm:main_variant_PT}} & 
      $\SigPtwo$-c. 		& {\mini\ref{thm:main_variant_PT}} \\
      $\MF$ & 
      $\in \sL$  				& {\mini\ref{thm:main_variant_F}} &
      $\in \sL$  				& {\mini\ref{thm:main_variant_F}} &
      $\in \sL$  				& {\mini\ref{thm:main_variant_F}} &
      $\in \P$ 					& {\mini\ref{thm:main_variant_F}} &
      $\SigPtwo$-c.  		& {\mini\ref{prop:abd(MF)-SigP2-c}} &
      $\SigPtwo$-c.  		& {\mini\ref{prop:abd(MF)-SigP2-c}} \\

  \end{tabularx}
  \smallskip
  \caption{%
    The complexity of $\ABD$, where 
    $*$-subscripts on clones denote all valid completions,
    $\sL$ abbreviates $\L$, and
    the suffix ``-c.'' indicates completeness for the respective complexity class.
  }
\end{table}

\begin{table}[H]
\setlength{\tabcolsep}{3pt}
\renewcommand{\arraystretch}{1.2}
  
  \begin{tabularx}{\linewidth}{X|c|c@{\;}c|c@{\;}c|c@{\;}c|c@{\;}c|c@{\;}c}
    Manifestation										 	&
      $\CloneE_*,\CloneN_*,$	&
      \multicolumn{2}{c|}{$\CloneL_1,\CloneL_2$} 	&
      \multicolumn{2}{c|}{$\CloneL_0,\CloneL_3,\CloneL$} &
      \multicolumn{2}{c|}{$\CloneD_2,\CloneS_{* 0}\subseteq$}	&
      \multicolumn{2}{c|}{$\CloneD_1,\CloneS_{* 2}\subseteq$} &
      \multicolumn{2}{c}{$\CloneD,\CloneS_1\subseteq$} \\
                            &
      $\CloneV_*$           &
      \multicolumn{2}{c|}{} &
      \multicolumn{2}{c|}{} &
      \multicolumn{2}{c|}{$[B]\subseteq\CloneM$}	&
      \multicolumn{2}{c|}{$[B]\subseteq\CloneR_1$} &
      \multicolumn{2}{c}{$[B]\subseteq\CloneBF$}\\    	
    \hline
    $\MNQ,\MNC$ 				& {$\in \sL$} 	&
      {$\in \sL$} 			& &
      {$\in \P$}				& {\mini\ref{rem:p-abd-Q-PQ-NQ}}	&
      {$\in \sL$}				& &
      {$\in \sL$}				& &
      {$\SigPtwo$-c.}		& {\mini\ref{rem:p-abd-Q-PQ-NQ}}	\\
    $\MNT$							& {$\in \sL$} 	&
      {$\in \sL$} 			& &
      {$\NP$-c.}				& {\mini\ref{rem:p-abd-T-PT-NT}}	&
      {$\in \sL$} 			& &
      {$\in \sL$}				& &
      {$\SigPtwo$-c.}		& {\mini\ref{rem:p-abd-T-PT-NT}}	\\
    $\MPQ,\MPC,\MQ,\MC$	& {$\in \sL$} 	&
      {$\in \P$}				& {\mini\ref{thm:P-abd(PQ)_main}}	&
      {$\in \P$}				& {\mini\ref{thm:P-abd(PQ)_main}}	&
      {$\in \sL$}		 		& {\mini\ref{prop:P-abd(PQ)-M-L}}	&
      {$\co\NP$-c.}			& {\mini\ref{prop:P-abd(PQ)-R1-coNP-c}}	&
      {$\SigPtwo$-c.}		& {\mini\ref{thm:P-abd(PQ)_main}}	\\
    $\MPT,\MT$					& {$\in \sL$} 	&
      {$\in \P$}				& {\mini\ref{thm:P-abd(PT)_main}}	&
      {$\NP$-c.} 				& {\mini\ref{prop:P-abd(PT)-L0-L3-NP-c}}	&
      {$\in \sL$}		 		& {\mini\ref{thm:P-abd(PT)_main}}	&
      {$\co\NP$-c.}			& {\mini\ref{thm:P-abd(PT)_main}}	&
      {$\SigPtwo$-c.}		& {\mini\ref{thm:P-abd(PT)_main}}	\\
    $\MF$								& {$\in \sL$} 		&
      {$\in \P$}				& {\mini\ref{thm:P-abd(F)_main}}	&
      {$\in \NP$}				& {\mini\ref{prop:P-abd(MF)-L-NP}} &
      {$\co\NP$-c.}			& {\mini\ref{prop:P-abd(MF)-M-coNP-c}}	&
      {$\co\NP$-c.}			& {\mini\ref{thm:P-abd(F)_main}}	&
      {$\SigPtwo$-c.}		& {\mini\ref{thm:P-abd(F)_main}}	\\

  \end{tabularx}
  \smallskip
  \caption{
    The complexity of $\ABD[\HP]$, where 
    $*$-subscripts on clones denote all valid completions,
    $\sL$ abbreviates $\L$, and
    the suffix ``-c.'' indicates completeness for the respective complexity class.
  }
\end{table}

Our results show, for instance, that when the knowledge base's formulae are restricted to be represented as positive clauses (\emph{i.e.}, $[B]=\CloneV$), then the abduction problem for single literal manifestations is very easy (solvable in $\L$); this still holds if the manifestations are represented by positive clauses. But its complexity jumps to $\NP$-completeness if we change the restriction on the manifestations to allow for positive terms.

Considering the case that all monotonic functions can be simulated (\emph{i.e.}, $X \subseteq [B] \subseteq \CloneM$ for $X \in \{\CloneD_2,\CloneS_{00},\CloneS_{10}\}$), the abduction problem is $\NP$-complete for manifestations represented by literals, clauses, or terms. Here allowing manifestation represented by a monotonic formulae, causes the jump to $\SigPtwo$-completeness. 
This increase in the complexity of the problem can be intuitively explained as follows.
The complexity of the abduction rests on two sources: finding a candidate explanation and checking that it is indeed a solution. 
The $\NP$-complete cases that occur in our classification hold for problems in which the verification can be performed in polynomial time. 
If both the knowledge base and the manifestation are represented by monotonic formulae, verifying a candidate explanation is $\co\NP$-complete.

It comes as no surprise that the complexity of $\ABD[\HP](B,\manif)$ is lower than or equal to the complexity of $\ABD(B,\manif) $ in most cases 
(except for the affine clones). We have seen in Lemma~\ref{lem:p-abd-useful} that 
for monotonic or $\true$-reproducing knowledge bases only one candidate needs to be considered. 
In these cases the complexity of the abduction problem is determined by the verification of the candidate.
This explains the appearance of $\coNP$-complete cases in our classification.
For the affine clones (\emph{i.e.}, $[B] \in \{\CloneL,\CloneL_0,\CloneL_3\}$), on the other hand,
the tractability of $\ABD(B,\manif)$ relies on Gaussian elimination. 
This method fails when restricting the hypotheses to be positive, and there is no obvious alternative.

\section{Counting complexity}\label{sec:counting}

  We now turn to the complexity of counting. We focus on the case where the manifestation is a single positive literal. For the symmetric abduction
  problem we are interested in counting the number of full explanations, $\#\ABD(B,\MPQ)$. For positive abduction two counting problems commonly arise: either to count all positive explanations, denoted by $\#\ABD[\HP](B,\MPQ)$; or to count only the subset-minimal explanations, denoted by $\#\text{-}\mathord{\subseteq}\text{-}\ABD[\HP](B,\MPQ)$. 
  Our first result in this section is the complete classification of $\#\ABD(B,\MPQ)$.
  
  
  

  \begin{theorem}\label{thm:counting_main}
    Let $B$ be a finite set of Boolean functions. Then the counting problem of symmetric abduction for propositional $B$-formulae with a positive literal manifestation, $\#\ABD(B,\MPQ)$, is
    \begin{enumerate}
      \item $\SHcoNP$-complete if $\CloneS_{02} \subseteq [B]$ or $\CloneS_{12} \subseteq [B]$ or	$\CloneD_{1} \subseteq [B]$,
      \item $\SHP$-complete if $\CloneV_{2}\subseteq [B]\subseteq \CloneM$ or $\CloneS_{10}\subseteq [B]\subseteq \CloneM$ or
            $\CloneD_{2}\subseteq [B] \subseteq \CloneM$,
      \item in $\FP$ in all other cases.
    \end{enumerate}
  \end{theorem}
  
  \begin{proof} 
    The $\SHcoNP$-membership for $\#\ABD(B,\MPQ)$ follows from the facts that checking whether a set of literals is indeed an explanation for an
    abduction problem is in $\P^\NP=\DeltaPtwo$ and from the equality $\SHclass{\DeltaPtwo}=\SHcoNP$, see \cite{hevo95}.
  
    We show $\SHcoNP$-hardness by giving a parsimonious reduction from the following $\SHcoNP$-complete problem:
    Count the number of satisfying assignments of $\psi(x_1,\ldots,x_n)=\forall y_1\cdots \forall y_m \varphi(x_1,\ldots,x_n,y_1,\ldots,y_m)$, where
    $\varphi$ is a DNF-formula (see, \emph{e.g.}, \cite{dhk05}).
    Let $x'_1,\ldots ,x'_n,  r_1,\ldots ,r_n,t$ and $q$ be fresh, pairwise distinct propositions. 
    We define the propositional abduction problem $\calP=(\Gamma, A, q)$ as follows:
    \begin{align*}
      \Gamma :=&\; \{x_i\rightarrow r_i, x'_i\rightarrow r_i, \neg x_i \vee \neg x_i' \mid 1\le i\le n\} \\
      \cup&\;\textstyle \{\varphi\rightarrow t\} \cup \{\bigwedge_{i=1}^n r_i\land t \rightarrow q\}, \\
      A :=&\; \{x_1,\ldots ,x_n\}\cup\{x'_1,\ldots ,x'_n\}.
    \end{align*}
    Observe that the manifestation $q$ occurs only in the formula $\bigwedge_{i=1}^n r_i\land t \rightarrow q$.
    This together with the formulae $x_i\rightarrow r_i, x'_i\rightarrow r_i, \neg x_i \vee \neg x_i'$, $1 \leq i \leq n$, enforces that every full
    explanation of $\calP$ has to select for each $i$ either $x_i$ and $\neg x'_i$, or  $\neg x_i$ and $x'_i$. By this the value of $x'_i$ is fully
    determined by the value of $x_i$ and is its dual. Moreover, it is easy to see that there is a one-to-one correspondence between the models of 
    $\psi$ and the full explanations of $\calP$.  

    Observe that since the reductions in Lemma~\ref{lem:constants_sometimes_available} are parsimonious, we can suppose w.l.o.g.\ that $B$ contains the
    two constants $\true$ and $\false$. Therefore, 
    it suffices to consider the case $[B] = \CloneBF$.
    Suppose that $\varphi = \bigvee_{i \in I} t_i$ and 
    let $\Gamma'$ be the set of formulae obtained by replacing 
    $\varphi \to t \equiv \neg \big(\bigvee_{i \in I} t_i\big) \lor t$ 
    by the set of clauses $\{\neg t_i \lor t | i \in I\}$.
    Then $\Gamma'$ is a set of disjunctions of literals, whose size is polynomially bounded by $|\Gamma|$.
    Hence, by the associativity of $\lor$, $\Gamma'$ can be transformed in logarithmic space into an 
    equivalent set of $B$-formulae. 
    This provides a  parsimonious reduction from the above $\SHcoNP$-complete problem to $\#\ABD(B,\MPQ)$.

    Let us now consider the $\SHP$-complete cases.
    When $[B]\subseteq \CloneM$, checking whether a set of literals $E$ is an explanation for an abduction problem with $B$-formulae is 
    in $\P$ (see Proposition~\ref{prop:abd(PQ)-M-NP-c}). This proves membership in $\SHP$. For the hardness result, it suffices to consider the case
    $[B]=\CloneV_{2}$, because the reduction provided in
    Lemma~\ref{lem:constants_sometimes_available} is parsimonious and $\CloneV_2 \subseteq [\CloneS_{10}\cup\{1\}]$.
    We provide a Turing reduction from the problem $\#\POSTWOSAT$, which is known to be $\SHP$-complete \cite{val79}. 
    Let $\varphi =\bigwedge_{i=1}^k ( p_i\lor  q_i)$ be an instance of this problem, 
    where $p_i$ and $q_i$ are propositional variables from the set $X=\{x_1,\ldots ,x_n\}$. 
    Let $q$ be a fresh proposition. Define the propositional abduction problem $\calP=(\Gamma, A, q)$ as follows:
    \[
      \Gamma := \{p_i\lor q_i\lor q\mid 1\le i\le k  \}, \qquad A:= \{x_1,\ldots ,x_n\}.
    \]
    It is easy to check that the number of satisfying assignments for $\varphi$ is equal to $2^n-\nSols{\calP}$.
    Finally, since $\lor\in[B]=\CloneV_{2}$, 
    $\calP$ can easily be transformed in logarithmic space into an $\ABD(B,\MPQ)$-instance.

    As for the tractable cases, the clones  $\CloneE$ and $\CloneN$ are easy; 
    and finally, for $[B]\subseteq \CloneL$, the number of full explanations is polynomial time computable according to \cite[Theorem~8]{hepi07}.
    \qed
  \end{proof}

  Turning to positive abduction the $\SHP$-complete cases vanish, while for the $\CloneL$-clones the exact complexity remains open.

  \begin{theorem}\label{thm:counting-p-abd}
    Let $B$ be a finite set of Boolean functions. 
    Then both counting problems of positive abduction for propositional $B$-formulae with a positive literal manifestation,
    $\#\ABD[\HP](B,\MPQ)$, $\#\text{-}\mathord{\subseteq}\text{-}\ABD[\HP](B,\MPQ)$ are
    \begin{enumerate}
      \item $\SHcoNP$-complete if $\CloneS_{02} \subseteq [B]$ or $\CloneS_{12} \subseteq [B]$ or	$\CloneD_{1} \subseteq [B]$,
      \item in $\SHP$ if $\CloneL_2 \subseteq B \subseteq \CloneL$,
      \item in $\FP$ in all other cases.
    \end{enumerate}
  \end{theorem}
  \begin{proof}
    The $\SHcoNP$-membership follows analogously to the proof of Theorem~\ref{thm:counting_main}. Indeed, the same reduction as in the proof of 
    Theorem~\ref{thm:counting_main} works: there is an one-to-one correspondence between full explanations and purely positive explanations.
    Moreover, all explanations are incomparable and hence subset-minimal.

    For the affine case membership in $\SHP$ follows from the $\NP$-membership of the corresponding decision problem.
    
    The remaining cases are encompassed by $[B] \subseteq \CloneM$. In this case a
    slight strengthening of Lemma~\ref{lem:p-abd-useful} is easily seen: Let $(\Gamma, A, q)$ be an instance. Then $A$ is an explanation
    if and only if all subsets of $A$ are explanations. Hence for $\#\ABD[\HP](B,\MPQ)$ the number of solutions is either $0$ or $2^{|A|}$
    (all subsets), while for $\#\text{-}\mathord{\subseteq}\text{-}\ABD[\HP](B,\MPQ)$ it is either $0$ or $1$ (the empty set). 
    We obtain membership in $\FP$, because for monotonic formulae deciding whether $A$ is an explanation can be done in $\L$.
    \qed
  \end{proof}

  We note that for manifestations represented as terms, clauses, or $B$-formulae, 
  most of the classifications of the corresponding counting problems can be easily derived from the above results; 
  the exceptions to this are $\#\ABD(\CloneM,\MF)$ and some cases satisfying $[B \cup \{\true\}]=\CloneL$, 
  whose exact complexity remains open.

\section{Concluding Remarks}\label{sec:conclusion}

In this paper we studied the decision and counting complexity of symmetric and positive propositional abduction from a knowledge base being represented as sets of $B$-formulae, for all possible finite sets $B$ of Boolean functions. 
We gave a detailed picture of the complexity of abduction in considering restrictions on both manifestations and hypotheses. 
Thus our results highlight the sources of intractability, identify fragments of lower complexity, and may help to identify candidates for parameters in the study of parameterized complexity of abduction.

Our restrictions on the hypotheses covered only the symmetric and the positive case.
One can as well define \emph{negative} abduction, where explanations consist of negative literals only; 
or \emph{non-symmetric} abduction, where explanations are formed upon a given set of \emph{literals}, which is \emph{not} demanded to be
closed under complement (in contrast to $\ABD[\HS]$). However, results not mentioned herein indicate that 
the classifications of these variants are easily seen to be identical to $\ABD[\HS]$ (except for the $\CloneL$-clones).

It is worth noticing that, with the exception of the clones between $\CloneL_2$ and $\CloneL$, 
whenever the abduction problem is tractable for some clone, it is trivial.
In contrast, tractability for the clones between $\CloneL_2$ and $\CloneL$ relies on Gaussian elimination, 
which fails when we restrict explanations to be positive.
Determining the complexity of positive abduction for the clone $\CloneL$ 
with manifestations represented by $\CloneL$-formulae might hence prove to be a challenging task
(note that a similar case, the circumscriptive inference of an affine formula from a set of affine formulae, remained unclassified in \cite{thomas09}).


\newcommand{\etalchar}[1]{$^{#1}$}

\end{document}